\documentclass[review]{elsarticle}

\usepackage{lineno,hyperref}
\usepackage[english]{babel}
\usepackage{microtype}
\usepackage{amsmath, amssymb, amstext}
\usepackage{amsthm}
\usepackage{txfonts} % Fonts that more closely match Elsevier's
\usepackage{color}
\usepackage{subcaption}
\usepackage{graphicx}
\usepackage{paralist}
\usepackage{xspace}
\usepackage{wrapfig}

\modulolinenumbers[5]

\graphicspath{{figures/}}

\newcommand{\etal}{et al.\xspace}
\newcommand{\T}{\mathcal{T}}
\newcommand{\ith}{\ensuremath{\textsuperscript{th}}}
\newcommand{\gS}{\ensuremath{\mathcal{S}}\xspace}
\newcommand{\gI}{\ensuremath{\mathcal{C}_{\text{in}}}\xspace}
\newcommand{\gO}{\ensuremath{\mathcal{C}_{\text{out}}}\xspace}
\newcommand{\vI}{\ensuremath{v_{\text{in}}}\xspace}
\newcommand{\vO}{\ensuremath{v_{\text{out}}}\xspace}
\newcommand{\ddd}{,\ldots ,}
\newcommand{\mc}{\mathcal}
\newcommand{\qg}{quadrilateral graph}

\newcommand{\Anote}[1]{{ #1}}
\newcommand{\bluenote}[1]{}
\newcommand{\Abluenote}[1]{#1}
\newcommand{\postRefereeChanges}[1]{#1}
\newcommand{\prc}[1]{\postRefereeChanges{#1}}
\newcommand{\remove}[1]{{}}
\newcommand{\ignore}[1]{}

\newtheorem{lemma}{Lemma}
\newtheorem{theorem}{Theorem}
\newtheorem{corollary}{Corollary}

\newtheorem*{conj}{Orbit Conjecture}

\journal{Computational Geometry: Theory and Applications}

\begin{document}

\begin{frontmatter}

\title{Flipping Edge-Labelled Triangulations\footnote{\prc{We dedicate this paper to Ferran Hurtado, whose creativity in research and seminal work on flips continues to inspire us all.}}}

\author[carleton]{Prosenjit Bose}     \ead{jit@scs.carleton.ca}
\author[waterloo]{Anna Lubiw}         \ead{alubiw@uwaterloo.ca}
\author[waterloo]{Vinayak Pathak}     \ead{vpathak@uwaterloo.ca}
\author[carleton]{Sander Verdonschot} \ead{sander@cg.scs.carleton.ca}

\address[carleton]{School of Computer Science, Carleton University, Ottawa ON, Canada}
\address[waterloo]{School of Computer Science, University of Waterloo, Waterloo ON, Canada}

\begin{abstract}
Flips in triangulations have received a lot of attention over the past decades.
However, the problem of tracking where particular edges go during the flipping process has not been addressed.   
We examine this question by attaching unique labels to the triangulation edges.
We introduce the concept of the \emph{orbit} of an edge $e$, which is the set of all edges reachable from $e$ via flips.

We establish the first upper and lower bounds on the diameter of the flip graph in this setting.
Specifically, we prove tight $\Theta(n \log n)$ bounds for edge-labelled triangulations of $n$-vertex convex polygons and combinatorial triangulations, contrasting with the $\Theta(n)$ bounds in their respective unlabelled settings.
\prc{The $\Omega(n \log n)$ lower bound for the convex polygon setting might be of independent interest, as it generalizes lower bounds on certain sorting models.}
When simultaneous flips are allowed, the upper bound for convex polygons decreases to $O(\log^2 n)$, although we no longer have a matching lower bound.

Moving beyond convex polygons, we show that edge-labelled triangulated polygons with a single reflex vertex can have a disconnected flip graph.
This is in sharp contrast with the unlabelled case, where the flip graph is connected for any triangulated polygon.
For spiral polygons, we provide a complete characterization of the orbits.
This allows us to decide \prc{connectivity} of the flip graph of a spiral polygon in linear time.
We also prove an upper bound of $O(n^2)$ on the diameter of each connected component, which is optimal in the worst case. 
We conclude with an example of a non-spiral polygon whose flip graph has diameter $\Omega(n^3)$.
\end{abstract}


\end{frontmatter}

%\linenumbers

\section{Introduction}
\label{sec:introduction}

The problem of reconfiguring one triangulation to another is well-studied for combinatorial triangulations, triangulations of points sets, and triangulations of polygons. In all these settings, the basic reconfiguration operation is the \emph{flip} operation that removes one edge of the triangulation and adds another edge to obtain a new triangulation. Not every edge can always be flipped. The constraints on the edges involved in a flip depend on the setting. \postRefereeChanges{The geometric setting refers to a triangulation of a point set (a maximal set of non-crossing edges joining pairs of points) or polygon (a maximal set of non-crossing chords inside the polygon).} In both cases, a flip removes one diagonal of a convex quadrilateral and replaces it by the other diagonal. In the combinatorial setting, a triangulation is a maximal planar graph with the clockwise order of edges around each vertex specified. Here, a flip removes one edge and replaces it by the other diagonal of the resulting quadrilateral, \prc{as} long as the new edge is not already an edge of the triangulation. When we replace an edge $e$ by an edge $f$ with a flip, we say that $e$ \emph{flips to} $f$. 

Another way to view reconfigurations via flips is through the \emph{flip graph}, which has a vertex for each triangulation and an edge between two vertices if their corresponding triangulations differ by a single flip. The most important property of flips is that they can be used to reconfigure any triangulation into any other triangulation of the same point set, polygon, or vertex set---that is, the flip graph is connected.

Wagner~\cite{Wag36} initiated the study of flips in 1936, by proving that the flip graph is connected in the combinatorial setting. For points sets, connectivity of the flip graph is a consequence of Lawson's result~\cite{Law77} that any triangulation of a point set can be flipped to the Delaunay triangulation, which then acts as a ``canonical'' triangulation from which every triangulation can be reached. For triangulations of a polygon, the constrained Delaunay triangulation can be used in the same way to prove that the flip graph is connected~\cite{Bern-Eppstein}. In Section~\ref{sec:background}, we discuss these results in more detail, along with bounds on the number of flips required, i.e., bounds on the diameter of the flip graph.

Reconfiguration of triangulations via flips is of theoretical interest in the study of associahedra~\cite{STT88} and mixing~\cite{molloy1999mixing}. It also has more practical benefits in mesh generation and for finding triangulations that optimize certain quality measures~\cite{Bern-Eppstein,Edelsbrunner}. For a broader overview of the literature on flips, we refer the reader to a survey by Bose and Hurtado~\cite{BH09}.

Despite the extensive literature on flips in many different settings, to our knowledge, no one has specifically studied where individual edges move during the course of a sequence of flips in any setting. We say that an edge $f$ is \emph{reachable} from an edge $e$ if there is a sequence of edges $e, \ldots, f$, such that for every two consecutive edges $a$ and $b$ in the sequence, there exists a triangulation in which $a$ flips to $b$. We define the \emph{orbit} of an edge $e$ to be the set of all edges reachable from $e$. This gives rise to a very natural question: What is the orbit of an edge? Although every triangulation can be flipped to every other, the orbits of edges can be very different. For example, some polygons have a unique triangulation, in which case the orbits are all singletons. At the other extreme, a convex polygon has a single orbit containing all the diagonals. In the geometric setting, the orbits are exactly the connected components of a graph introduced by Eppstein~\cite{Epp} as the \emph{\qg}. This graph has a vertex for every diagonal and an edge between two vertices if their corresponding diagonals $e$ and $f$ intersect and their endpoints form a convex quadrilateral that is empty.  

An intriguing question is how orbits combine. In other words, we want to track where multiple edges go during a sequence of flips. To address this question, we study flips in \emph{edge-labelled triangulations}: triangulations of a point set, polygon, or vertex set where each edge has a unique label from $\{1, \ldots, m\}$, where $m$ is the number of edges. If we flip an edge of an edge-labelled triangulation, the new edge is assigned the label of the removed edge. In particular, this means that the set of edge labels is preserved throughout any flip sequence. In the geometric setting, we often omit the labels on the edges of the convex hull or on the boundary of the polygon, as these can never be flipped.

Given two edge-labelled triangulations $T$ and $T'$ of the same point set, polygon, or vertex set, it is not always possible to transform $T$ to $T'$ via a sequence of flips. For example, when a polygon has a unique triangulation, no flips are possible, so if the labellings of $T$ and $T'$ differ, no flip sequence can transform one to the other. A necessary condition for such a flip sequence to exist is that each edge with label $i$ in $T'$ must be in the orbit of the edge with label $i$ in $T$. We conjecture that this condition is also sufficient.

\begin{conj}
 \label{con:main}
 Given two edge-labelled triangulations, there is a sequence of flips that transforms one into the other if and only if for every label, the initial and final edge with that label lie in the same orbit.
\end{conj}

The Orbit Conjecture is known to hold in some cases. A result by Eppstein~\cite{Epp} implies that the Orbit Conjecture holds for triangulations of point sets with no empty pentagon. These point sets are highly degenerate since every set of ten points in general position contains an empty pentagon. Eppstein showed that in this case every orbit is a tree, and each triangulation contains exactly one edge from each orbit. There are a few other results of a similar flavour. Cano~\etal~\cite{JDHU13} proved the analogous statement for edge-labelled non-maximal plane graphs, although instead of flips, they ``rotate" edges around one of their endpoints. Hernando~\etal~\cite{HHMR03} proved the analogous statement for edge-labelled spanning trees of an underlying graph, where the ``flip" operation consists of removing an edge and replacing it elsewhere with the same label while maintaining connectivity. They showed that the orbits are exactly the 2-connected components of the underlying graph. Pathak~\cite{Pathak} generalized this result to labelled bases of a matroid.

\prc{Edge-labelled flips in triangulations of a convex polygon were independently studied by Araujo-Pardo~\etal~\cite{AHOS14}, who mainly focused on the combinatorial properties of the flip graph. In the unlabelled setting, the flip graph is known to be isomorphic to the 1-skeleton of a polyhedron, called the associahedron. Araujo-Pardo~\etal proved that this is also the case for the edge-labelled setting, and called the resulting polyhedron the colourful associahedron. As part of their proof, they show that the flip graph of edge-labelled triangulations of a convex polygon is connected, which implies the Orbit Conjecture for this setting. A careful analysis of their argument shows that it gives an upper bound on the diameter of the flip graph that is quadratic in the number of vertices of the polygon. In addition, they show that the colourful associahedron covers the regular associahedron and fully characterize its automorphism group.}

In this paper we address the Orbit Conjecture for edge-labelled triangulations of polygons and edge-labelled combinatorial triangulations. We give an alternative proof for the connectivity of the flip graph of edge-labelled triangulations of an $n$-vertex convex polygon that shows that the diameter is $O(n \log n)$, and we provide a matching $\Omega(n \log n)$ lower bound, which contrasts with the $\Theta(n)$ bound for unlabelled flips~\cite{STT88}.

We also consider what happens when we allow multiple edges to be flipped simultaneously. We prove that the simultaneous flip distance between two edge-labelled triangulations of a convex polygon is $O(\log^2 n)$, in contrast with the $\Theta(\log n)$ bound established for the unlabelled setting~\cite{Galtier03}.

After convex polygons, we consider spiral polygons (polygons with at most one reflex chain). These may have multiple orbits, so the flip graph is not necessarily connected. We again prove the Orbit Conjecture, and give a tight $\Theta(n^2)$ worst case bound on the diameter of each connected component of the flip graph. We also characterize the orbits, and show how to test in linear time whether 
there is a flip sequence that transforms one given edge-labelled triangulation into another. 

Turning to more general polygons, we give an example of a polygon with two reflex chains whose flip graph is connected but has diameter $\Omega(n^3)$. This is in stark contrast to the $\Theta(n^2)$ bound on the diameter of the flip graph in the unlabelled setting~\cite{Bern-Eppstein,HNU99}, and the $\Theta(n^2)$ bound for spiral polygons.

Finally, we consider the case of $n$-vertex edge-labelled combinatorial triangulations. As for convex polygons, we prove the Orbit Conjecture by showing that the flip graph is connected and has diameter $\Theta(n \log n)$, contrasting with the $\Theta(n)$ bound for the unlabelled case~\cite{STT92}.

%%%%%%%%%%%%%%%%%%%%%%%%%%%%%%
\section{Background}
\label{sec:background}

The original result on flips was due to Wagner~\cite{Wag36}, who showed that $O(n^2)$ flips were sufficient to transform between any two given $n$-vertex combinatorial triangulations. In particular, he showed that any triangulation can be transformed into one with two dominant vertices, often referred to as \emph{Wagner's canonical triangulation}. The upper bound was first improved to linear by Sleator et al.~\cite{STT92}---this is optimal, since converting a triangulation with constant maximum degree into Wagner's canonical triangulation immediately gives a linear lower bound. Currently, the best upper bound on the diameter of the flip graph is $5n-23$ \cite{CHKTW15}. See Bose and Verdonschot~\cite{BV11} for a more detailed history of this problem in the combinatorial setting.

Combinatorial triangulations with vertex labels were studied by Sleator et al.~\cite{STT92} as an intermediate form between unlabelled combinatorial triangulations and triangulations of a set of points in the plane. In this setting, they showed that $O(n \log n)$ flips are sufficient to transform one vertex-labelled triangulation with $n$ vertices into any other, and that $\Omega(n \log n)$ flips are sometimes necessary. Note that by transforming both the initial and final triangulation into Wagner's canonical form without paying attention to vertex labels, the problem becomes one of sorting vertex labels. This is essentially what they do to prove the $O(n \log n)$ upper bound. For the $\Omega(n \log n)$ lower bound, they show that there are at most $9^{n+m}$ distinct vertex-labelled triangulations that are reachable from a given triangulation $T$ via $m$ flips. Since there are already $(n-3)!$ different ways to label the vertices of Wagner's canonical triangulation, the $\Omega(n \log n)$ lower bound follows.

For $n$-vertex triangulations of point sets, Lawson~\cite{Law77} proved that given any triangulation, a sequence of $O(n^2)$ flips suffices to transform it to the Delaunay triangulation. For triangulations of an $n$-vertex polygon, Bern and Eppstein~\cite{Bern-Eppstein} showed that the flip graph is also connected and has diameter $O(n^2)$. Unlike in the combinatorial setting, Hurtado~\etal~\cite{HNU99} proved that these quadratic bounds are tight by providing an $\Omega(n^2)$ lower bound.  

Flips in triangulations of a convex polygon are especially interesting because they correspond exactly to rotations in a binary tree. In this way, Sleator et al.~\cite{STT88} answered a question about the rotation distance between binary trees, by proving a tight bound of $2n-10$ on the diameter of the flip graph. Recently, Pournin~\cite{Pournin13} found a purely combinatorial proof of this bound, avoiding the use of hyperbolic geometry in the original proof and establishing the bound even for small values of $n$.

The idea of performing flips in parallel was introduced by Hurtado~\etal~\cite{Hurtado98}. In the geometric setting, a set of edges may be \emph{simultaneously flipped} if each edge is flippable and no two edges are incident to the same triangle. Galtier~\etal~\cite{Galtier03} showed that $\Theta(n)$ simultaneous flips are sufficient and sometimes necessary to reconfigure one triangulation of a point set or simple polygon to another. In the case of convex polygons, they showed that $\Theta(\log n)$ simultaneous flips were sufficient and sometimes necessary. Bose~\etal~\cite{Simultaneous06} proved that $\Theta(\log n)$ simultaneous flips are sufficient and sometimes necessary in the combinatorial setting---in this case a simultaneous flip may be performed even if some edges in the set cannot be individually flipped.

\section{Convex polygons}
\label{sec:convex}

\prc{In this section, we prove that the flip graph of edge-labelled triangulations of a convex polygon with $n$ vertices is connected. Since all diagonals are in the same orbit, this is equivalent to proving the Orbit Conjecture in this setting. Aurajo-Pardo~\etal~\cite{AHOS14} also proved this, but our proof gives a tight bound of $\Theta(n \log n)$ on the diameter of the flip graph.} In addition, we show that $O(\log^2n)$ simultaneous flips are sufficient and $\Omega(\log n)$ simultaneous flips are sometimes necessary.

\begin{figure}[htb]
\centering
\includegraphics{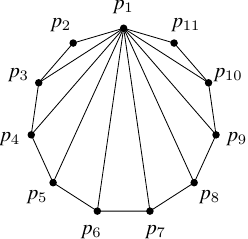}
\caption{A canonical triangulation with 11 vertices.}
\label{fig:canonical}
\end{figure}

Given a convex polygon with $n$ vertices, label the vertices $p_1, \ldots, p_{n}$ in counter-clockwise order and give each diagonal a unique label from the set $\{1, \ldots, n-3\}$.
Denote by $T^*$ the \emph{canonical} triangulation where vertex $p_1$ is dominant (see Figure~\ref{fig:canonical}).
Let $\rho$ be the sequence of edge labels ordered in counter-clockwise order around $p_1$, starting with diagonal $p_1p_3$.
Note that $\rho$ is a permutation of $[1..n-3]$.
We denote this canonical triangulation by $(T^*, \rho)$.
The next lemma shows that in this canonical triangulation, ordered subsequences can easily be reversed, which plays an important role for the upper bound.
\prc{Note the distinction between an \emph{ordered} subsequence and a \emph{contiguous} subsequence. An ordered subsequence is a sequence consisting of a subset of the elements of the original sequence in the same order as in the original. A contiguous subsequence is an ordered subsequence in which the elements are consecutive in the original sequence.}

\begin{figure}[htb]
\centering
\includegraphics{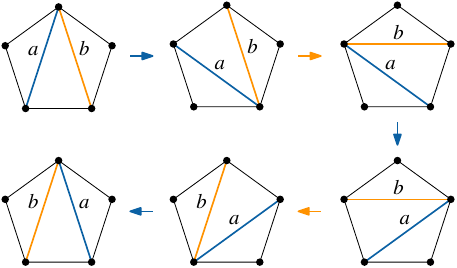}
\caption{A sequence of five flips that reverses the order of the two diagonals of a pentagon.}
\label{fig:pentagonswap}
\end{figure}

\begin{lemma}
\label{lemma:noncontiguous}
Given $(T^*, \rho)$, let $S$ be a contiguous subsequence of diagonals adjacent to vertex $p_1$, ordered in counter-clockwise order. Any ordered subsequence $R=i_1, \ldots , i_r$ of $S$ can be reversed with \prc{at most} $2.5|S|$ flips.
\end{lemma}
\begin{proof}
Note that the edges of $R$ need not be consecutive. We use induction on $|S|$. In the base case, $|S| = 0$ and we do not need any flips to reverse $R$, so let $|S| > 0$ and assume that the lemma holds for any contiguous subsequence $S'$ with $|S'| < |S|$.

If there is any edge $e$ in $S$ that is not in $R$, we flip that edge. Now the remaining edges of $S$ form a contiguous subsequence of diagonals in a smaller convex polygon. Thus, we can reverse them with $2.5(|S| - 1)$ flips by induction. Flipping $e$ back completes the transformation, for a total of $2.5(|S| - 1) + 2 \leq 2.5|S|$ flips.

Now consider the case where $S = R$. If $|R|$ is odd, we can remove the middle diagonal from $R$, bringing us back to the previous case, so assume that $|R|$ is even. Then the two middle diagonals, $i_{r/2}$ and $i_{r/2+1}$, can be swapped with five flips, as shown in Figure~\ref{fig:pentagonswap}. But instead of completing this swap sequence, we halt it after the first three flips. This leaves the edges out of the way of the remaining edges, which we then reverse with $2.5(|S| - 2)$ flips by induction. Afterwards, we perform the final two flips to return $i_{r/2}$ and $i_{r/2+1}$, bringing the total number of flips to $2.5(|S| - 2) + 5 = 2.5|S|$.
\end{proof}

Using Lemma \ref{lemma:noncontiguous}, we can sort $\rho$ with $O(n \log n)$ flips, which gives us the following theorem.
\begin{theorem}
\label{thm-upper-bound}
Any edge-labelled triangulation of a convex polygon %$P$
with $n$ vertices can be transformed into any other edge-labelled triangulation using $O(n \log n)$ flips.
\end{theorem}
\begin{proof}
We first ignore the labels and use $O(n)$ flips to transform \prc{both initial triangulations to the canonical triangulations} $(T^*, l'_1)$ and $(T^*, l'_2)$, respectively%~\cite{STT88}
.
Next, we show that for any permutation $\rho$, the canonical triangulation $(T^*, \rho)$ can be transformed into the triangulation $(T^*, \rho^*)$ using $O(n\log n)$ flips, where $\rho^*$ is the sorted permutation.
%Next, we show that a canonical triangulation $(T^*, \rho)$ corresponding to any permutation $\rho$ can be transformed using $O(n\log n)$ flips into the triangulation $(T^*, \rho^*)$, where $\rho^*$ is the sorted permutation.
 
We sort the permutation $\rho$ by imitating a version of quicksort that always uses the median as pivot. This ensures that the $O(n \log n)$ bound holds even in the worst case.
%We can sort in this model by imitating quicksort but in a balanced way to achieve $O(n \log n)$.
\prc{Let $\rho_i$ denote the $i$-th element in $\rho$ and} consider all $i \in [1 .. n-3]$ for which either: (1) $i\leq (n-3)/2$ and $\rho_i>(n-3)/2$; or (2) $i> (n-3)/2$ and $\rho_i\leq (n-3)/2$. Reverse the subsequence formed by these values of $i$. By Lemma~\ref{lemma:noncontiguous}, this can be done with $O(n)$ flips. This ensures that edges whose label is below $(n-3)/2$ lie in the first half of the sequence and the larger labels lie in the second half. By recursively applying this operation to the two halves, we sort the entire sequence using $O(n \log n)$ flips.
\end{proof}

As a corollary to the $O(n\log n)$ upper bound, we can get an $O(\log n)$-factor approximation algorithm for computing the flip-distance.

\begin{corollary}
\label{cor:approx}
We can approximate in polynomial time the flip-distance between two $n$-vertex edge-labelled triangulations $\T_1$ and $\T_2$ \prc{of a convex polygon} within an $O(\log n)$ factor \prc{of optimal}.
\end{corollary}
\begin{proof}
Call an edge $e$ \emph{fixed} if it satisfies the following properties:
\begin{inparaenum}[\itshape a\upshape)]
\item it occurs in both $\T_1$ and $\T_2$;
\item it has the same label in both $\T_1$ and $\T_2$; and
\item the set of distinct labels that occurs to the left of $e$ in $\T_1$ is exactly the same as the set that occurs to the left of $e$ in $\T_2$.
\end{inparaenum}
%Our algorithm never flips a fixed edge. % SV: this is not necessary and easily misunderstood to refer to the algorithm of Thm 1, which would be false.
Since
%Note that 
every non-fixed edge must %flip
be flipped
at least once%.
, we need at least $k$ flips, where $k$ is the number of non-fixed edges.
The set of fixed edges divides the polygon into several smaller convex polygons, each of which 
%is triangulated differently 
has a different edge-labelled triangulation
in $\T_1$ and $\T_2$. 
Suppose the $i\ith$ polygon has $n_i$ diagonals.
%These are non-fixed so we have a lower bound of $\Omega(n_i)$ flips.
We use the %$O(n_i\log n_i)$ 
flip-sequence from Theorem~\ref{thm-upper-bound} on the $i\ith$ polygon
to perform the required transformation with $O(n_i\log n_i)$ flips. The total number of flips is then $\sum_i O(n_i\log n_i) \leq O(k \log k)$
, thus giving an approximation factor of $O(\log n)$.
\end{proof}

\prc{
\subsection{Lower bound}
\label{sec:el-convex-lb}

The lower bound uses a slightly modified version of the $\Omega(n \log n)$ lower bound for the vertex-labelled setting by Sleator, Tarjan, and Thurston~\cite{STT92}. We first give an overview of their technique, before applying it to edge-la\-belled triangulations of a convex polygon.

Let a \emph{tagged half-edge graph} be an undirected graph with maximum degree $\Delta$, whose vertices have labels called \emph{tags}, and whose edges are split into two \emph{half-edges}. Each half-edge is incident to one endpoint, and labelled with an \emph{edge-end label} in $\{1, \ldots, \Delta\}$, such that all edge-end labels incident on a vertex are distinct (see Figure~\ref{fig:el-half-edge-graph} for an example). A \emph{half-edge part} is a half-edge graph in which some half-edges do not have a twin. Note that tags are not restricted to integers: they could be tuples, or even arbitrary strings.

\begin{figure}[htb]
 \centering
 \begin{subfigure}[b]{0.4\textwidth}
  \centering
  \includegraphics{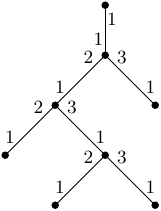}
  \caption{}
  \label{fig:el-half-edge-graph}
 \end{subfigure}
 \begin{subfigure}[b]{0.58\textwidth}
  \centering
  \includegraphics{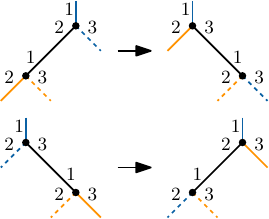}
  \caption{}
  \label{fig:el-rotation-grammar}
 \end{subfigure}
 \caption{(a) A half-edge graph representation of a rooted binary tree. (b) A graph grammar for rotations in binary trees. Correspondence between half-edges is indicated by a combination of colour and line style.}
\end{figure}

A \emph{graph grammar} $\Gamma$ is a sequence of production rules $\Gamma_i = (L, \rightarrow, \mathcal{T}, R)$, where $L$ and $R$ are half-edge parts with the same number of half-edges, $\rightarrow$ is a correspondence between the half-edges of $L$ and $R$, and $\mathcal{T}$ is a function that computes the tags of vertices in $R$ from those in $L$. A possible graph grammar for rotations in (unlabelled) binary trees is depicted in Figure~\ref{fig:el-rotation-grammar}.

Sleator, Tarjan, and Thurston prove the following theorem.

\begin{theorem}[Sleator, Tarjan, and Thurston~\cite{STT92}]
 \label{thm:el-stt-reachable}
 Let $G$ be a tagged half-edge graph with $n$ vertices, $\Gamma$ be a graph grammar, $c$ be the number of vertices in left sides of $\Gamma$, and $r$ be the maximum number of vertices in any right side of a production of $\Gamma$. Then $|R(G, \Gamma, m)| \leq (c + 1)^{n + r \cdot m}$, where $R(G, \Gamma, m)$ is the set of graphs obtainable from $G$ by derivations in $\Gamma$ of length at most $m$.
\end{theorem}

Recall that, by definition, the maximum degree of a tagged half-edge-graph is bounded by a constant. Thus, we cannot apply this theorem directly to triangulations of a convex polygon. Instead, we turn to the dual graph. The \emph{augmented dual graph} of a triangulation of a convex polygon is a tagged half-edge graph $G$ with two sets of vertices: triangle-vertices $T$ corresponding to the triangles of the triangulation, and edge-vertices $E_{\textsc{CH}}$ corresponding to the boundary edges. One edge-vertex is designated as the \emph{root}.

Two triangle-vertices are connected by an edge if their triangles are adjacent. All edge-vertices are leaves, each connected to the triangle-vertex whose triangle is incident to their corresponding edge (see Figure~\ref{fig:el-triangulation-to-tree}). As every triangle has three edges, the maximum degree of $G$ is three. The edge towards the root receives edge-end label $1$. For a triangle-vertex, the other edge-end labels are assigned in counter-clockwise order, as in Figure~\ref{fig:el-half-edge-graph}.

\begin{figure}[htb]
 \centering
 \includegraphics{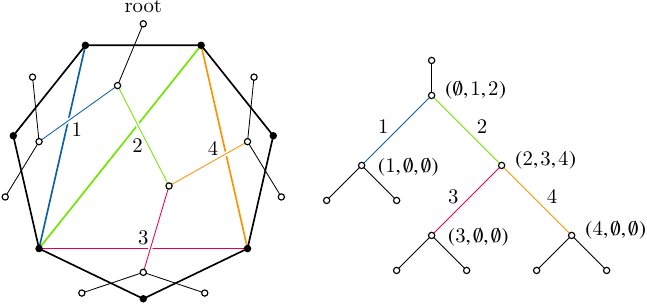}
 \caption{An edge-la\-belled triangulation of a convex polygon with its augmented dual graph. The \prc{edge-labels} on the dual graph are shown to more clearly indicate the correspondence -- they are actually labelled with edge-end labels as in Figure~\ref{fig:el-half-edge-graph}.}
 \label{fig:el-triangulation-to-tree}
\end{figure}

This is where we deviate slightly from the original paper. Since Sleator, Tarjan, and Thurston were working in the vertex-labelled setting, they used the tags in the augmented dual graph to encode the labels of the vertices around the corresponding triangles. Instead, we use these tags to encode the edge-labels. Specifically, we tag each triangle-vertex with a triple containing the edge-label of each edge of its triangle, starting from the edge closest to the root, and proceeding in counter-clockwise order. Edges of the convex hull are assumed to have label $\emptyset$. Edge-vertices will not be involved in any of the production rules, so they do not need tags.

As flips in the triangulation correspond to rotations in the augmented dual graph~\cite{STT88}, the graph grammar is identical to the graph grammar presented before. The only addition is the computation of new tags for the vertices on the right-hand side (see Figure~\ref{fig:el-labelled-rotation-grammar}). This grammar has four vertices in left sides, and a maximum of two vertices in any right side. Since a triangulation of an $n$-vertex convex polygon has $n - 2$ triangles and $n$ convex hull edges, the augmented dual graph has $2n - 2$ vertices. Thus, Theorem~\ref{thm:el-stt-reachable} gives us the following.

\begin{figure}[htb]
 \centering
 \includegraphics{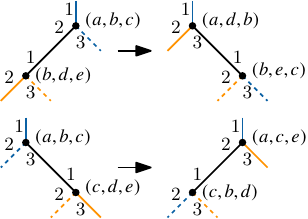}
 \caption{A graph grammar for rotations in augmented dual graphs, which correspond to flips in edge-la\-belled triangulations of a convex polygon.}
 \label{fig:el-labelled-rotation-grammar}
\end{figure}

\begin{lemma}
 \label{lem:el-convex-reachable}
 Given an edge-la\-belled triangulation $G$ of an $n$-vertex convex polygon, the number of distinct edge-la\-belled triangulations reachable from $G$ in $m$ flips is at most $5^{2n - 2 + 2m}$.
\end{lemma}

This bound can be further refined to $3^{n - 1 + 2m}$, using the leader-follower and zero-elimination techniques from Sleator, Tarjan, and Thurton's paper~\cite{STT92}. However, the cruder bound already suffices to derive the correct asymptotic lower bound.

\begin{theorem}
 \label{thm:el-convex-lb}
 There are pairs of edge-la\-belled triangulations of a convex polygon with $n$ vertices such that transforming one into the other requires $\Omega(n \log n)$ flips.
\end{theorem}
\begin{proof}
 We first estimate the number of edge-la\-belled triangulations. An $n$-vertex convex polygon has $n - 3$ diagonals, and in a fan triangulation, each sequence of labellings results in a new triangulation. Thus, there are at least $(n - 3)!$ edge-la\-belled triangulations.

 Let $d$ be the diameter of the flip graph. Then, for every graph $G$, $d$ flips suffice to reach all edge-la\-belled triangulations. But from Lemma~\ref{lem:el-convex-reachable}, we know that a sequence of $m$ flips can generate at most $5^{2n - 2 + 2m}$ unique edge-la\-belled triangulations. This gives us the following bound.
 \begin{align*}
  5^{2n - 2 + 2d}~~&\geq~~(n - 3)! \\
  \log_5 5^{2n - 2 + 2d}~~&\geq~~\log_5 (n - 3)! \\
  2n - 2 + 2d~~&\geq~~\log_5 (n! / n^3) \\
  2d~~&\geq~~\log_5 n! - \log_5 n^3 - 2n + 2\\
  2d~~&\geq~~\Omega(n \log n) - O(n) \\
  d~~&\geq~~\Omega(n \log n) \qedhere
 \end{align*}
\end{proof}

In Section~\ref{sec:sorting}, we discuss the implications of this theorem on sorting permutations in length-weighted models.
}

\subsection{Simultaneous flips}

We now turn our attention to simultaneous flips. Recall that a set of edges may be simultaneously flipped if each edge is flippable and no two edges are incident to the same triangle. For edge-labelled triangulations of a convex polygon, we show that $O(\log^2 n)$ simultaneous flips suffice. Our approach essentially hinges on showing how to perform a balanced partition step of quicksort with $O(\log n)$ simultaneous flips.

\begin{theorem}\label{thm:simultaneous}
Given two edge-labelled triangulations of a convex polygon with $n$ vertices, $O(\log^2 n)$ simultaneous flips are sufficient to transform one to the other.
\end{theorem}
\begin{proof}
We first ignore the labels and, by the result of Galtier~\etal~\cite{Galtier03}, use $O(\log n)$ simultaneous flips to transform both \prc{edge-labelled triangulations} to the canonical triangulations $(T^*, l'_1)$ and $(T^*, l'_2)$, respectively.
Next, we show how to transform a general canonical triangulation $(T^*, \rho)$ to $(T^*, \rho^*)$, where $\rho^*$ is the sorted permutation.
We do this by imitating quicksort.
In particular, we show that the \prc{partition} step of quicksort can be carried out in $O(\log n)$ simultaneous flips.
Then, by recursing \prc{on} both halves simultaneously, we get that the total number of simultaneous flips satisfies the recursion $T(n)=T(n/2)+O(\log n)$, which solves to $T(n) = O(\log^2 n)$.

We now address the partition step of quicksort.
Let $E$ be the set of all non-boundary edges, and let $E_l=\{\rho_i \mid i < n/2$ and $\rho_i\geq n/2\}$ be the set of edges that are in the left half in $\rho$ but in the right half in $\rho^*$.
Similarly, let $E_r$ be the set of edges that are in the right half in $\rho$ but in the left half in $\rho^*$.
We first flip all the edges of $E \setminus (E_l \cup E_r)$.
Since one simultaneous flip can flip every other edge of the set, $O(\log n)$ simultaneous flips suffice to flip the whole set.
The edges of $(E_l \cup E_r)$ now form a canonical triangulation of a smaller convex polygon.
For the remainder of the proof we work only with this smaller convex polygon.  
In other words, we only consider canonical triangulations $(T^*, \rho)$ where \emph{every} edge on the left half wants to go to the right half and vice versa.

Let $\T$ be a canonical triangulation where every edge on the left half is coloured red and every edge on the right half is coloured blue.
Let $\T'$ be the triangulation with the colours interchanged.
We show how to transform $\T$ to $\T'$ with $O(\log n)$ simultaneous flips.
We do this by converting both $\T$ and $\T'$ into a common coloured triangulation in which the diagonals form a path that alternates between red and blue edges, as shown in Figure~\ref{fig:alternating-zigzagX}. 
We call this triangulation an \emph{alternating zig-zag}.
Note that $p_1$ is the high degree vertex of the canonical triangulation.

\begin{figure}[htb]
\centering
 \begin{subfigure}[b]{0.48\textwidth}
  \centering
  \includegraphics{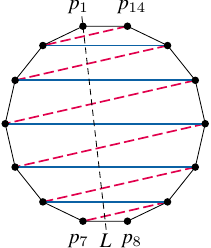}
  \caption{}
  \label{fig:alternating-zigzagX}
 \end{subfigure}
 \begin{subfigure}[b]{0.48\textwidth}
  \centering
  \includegraphics{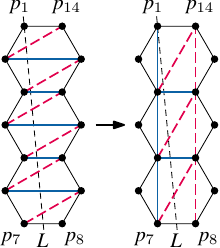}
  \caption{}
  \label{fig:alternating-zigzag-hexagons}
 \end{subfigure}
 \caption{(a) An 11-edge zig-zag alternating between red (dotted) and blue (solid) edges.
(b) An alternating zig-zag can be thought of as a collection of hexagons. In $O(1)$ simultaneous flips, we can move half of the edges to the correct side of $L$, with the rest forming a smaller alternating zig-zag.}
\end{figure}

Imagine a line $L$ passing through $p_1$ that separates the red edges from the blue edges in the canonical triangulation and crosses all the edges in the alternating zig-zag. 
We show how to transform the alternating zig-zag to $\T'$.
We claim that with $O(1)$ simultaneous flips, we can ensure that half the blue edges are to the left of this line and half the red edges are to the right.
Moreover, the edges that still cross $L$ form an alternating zig-zag of a subpolygon at most half the size of $\T'$.

To do this, note that each four consecutive triangles of the zig-zag form a hexagon, as shown in Figure~\ref{fig:alternating-zigzag-hexagons}.
Each hexagon has three diagonals, two red and one blue, and the hexagons are separated by the remaining blue edges.
As any transformation of a hexagon takes a constant number of flips, we can reconfigure all hexagons with a constant number of simultaneous flips.
In particular, we transform each hexagon to have one blue edge to the left of $L$ and one red edge to the right of $L$, as shown in Figure~\ref{fig:alternating-zigzag-hexagons}.
The triangles that still intersect $L$ now form an alternating zig-zag of half the size, consisting of the remaining red diagonal inside each hexagon and the blue edges separating the hexagons.
Thus, by repeating the process $O(\log n)$ times, we ensure that every blue edge is to the left of $L$ and every red edge is to the right.

Now the subpolygon to the left of $L$ has only blue diagonals, all of which we can make incident to $p_1$ with $O(\log n)$ simultaneous flips.
By doing the same with the red edges, we can transform the alternating zig-zag into $\mathcal{T'}$ using $O(\log n)$ simultaneous flips.
Similarly, we can transform the alternating zig-zag into $\mathcal{T}$ by moving the blue diagonal of each hexagon to the right and one of the red diagonals to the left.
Thus we can perform the partition step in $O(\log n)$ simultaneous flips, which completes the proof.
\end{proof}

The $\Omega(\log n)$ lower bound on the worst-case number of simultaneous flips in the unlabelled case trivially provides an $\Omega(\log n)$ lower bound in the edge-labelled setting as well. 
We prove a stronger lower bound. In particular, the following theorem shows that even the ``partition'' step of quick sort requires at least $\Omega(\log n)$ simultaneous flips.
 
\begin{theorem}
\label{thm-simlb}
Let $\T$ be a canonical triangulation on $n$ vertices with all diagonals in the left half coloured red, and all edges in the right half coloured blue. Let $\T'$ be the same canonical triangulation, with the colours interchanged. Transforming $\T$ to $\T'$ requires at least $\Omega(\log n)$ simultaneous flips.
\end{theorem}
\begin{proof}
Consider the line $L$ that passes through $p_1$ and separates the red edges from the blue ones in $\T$. For each edge, the side of $L$ it inhabits in $\T$ is different from the side it inhabits in $\T'$. A single flip cannot replace an edge that lies completely to the left of $L$ with an edge that lies completely to the right. Thus, in any simultaneous flip-sequence that transforms $\T$ to $\T'$, for any edge $e$, there must be a triangulation where $e$ intersects $L$.

Consider any simultaneous flip-sequence $\mathcal{F}$ that transforms $\T$ to $\T'$, i.e., $\mathcal{F}$ is the sequence $(\T = \T_1, \T_2,\ldots , \T_{k-1}, \T_k=\T')$. For any $j\in [1..k]$, let $c_j$ be the number of edges $e$ for which there exists a $\T_r$ with $r\leq j$ such that $e$ intersects $L$ in $T_r$. From the argument above, it is clear that $c_1 = 0$ and $c_k = n$. We claim that for all $j\in [1..k-1]$, $c_{j+1} \leq 2c_j+1$. This shows that $k\geq \Omega(\log n)$.

To see why the claim is true, consider the $j\ith$ simultaneous flip in the sequence. It makes $\Delta_j = c_{j+1}-c_j$ new edges cross $L$.
Each individual flip happens in its own quadrilateral that has exactly two boundary-edges that intersect $L$.
%Each flip in the simultaneous flip happens in its own quadrilateral and exactly two boundary-edges of each quadrilateral intersect $L$.
Since two adjacent quadrilaterals can share an edge and since the quadrilateral at the top and bottom share their boundaries with the polygon, the total number of quadrilateral edges that cross $L$ in $\T_j$ is at least $\Delta_j-1$. But this number is also at most $c_j$. Therefore $\Delta_j-1 \leq c_j$, which means $c_{j+1} \leq 2c_j+1$.
\end{proof}

%%%%%%%%%%%%%%%%%%%%%%%%%%%%%%
\section{Connections between flips and sorting models}
\label{sec:sorting}

Our upper bounds in the previous section depend on bounds for sorting permutations.
In this section, we study the connection in the other direction, showing that our lower bound on the diameter of the flip graph implies an $\Omega(n \log n)$ lower bound on the cost of sorting in \prc{various} length-weighted sorting models.

Sorting permutations of $[1..n]$ using %operations from 
a restricted set of operations has been widely studied. 
A main operation is the reversal of a subsequence.
One of the earliest results of this kind was on ``pancake'' sorting~\cite{GP79}, where a prefix of the sequence can be reversed.
Bubble sort also fits into this model, as it operates by swapping two adjacent elements at a time, which is a reversal of size two. The number of size-two reversals it makes is equal to the number of inversions of the permutation and is $\Theta(n^2)$
%~\cite{CLR01} 
in the worst case. More generally, any permutation can be sorted with $\Theta(n)$ reversals~\cite{KS97} if arbitrarily large reversals are allowed. 
%Computing the minimum number of operations (of a specific kind) required to sort a permutation is of interest in bioinformatics. In the case of arbitrarily large reversals, this problem is NP-complete~\cite{Carprara03} for sorting permutations, but is in P for sorting \emph{signed} permutations~\cite{BMY01}. 
%Some other kinds of operations that have been considered include transpositions and block interchanges~\cite{DRS10}.
Most of the results until now have been about the \emph{number} of operations required. Recently, Pinter and Skiena~\cite{PS02} formulated a model that 
allows reversals of any size, but assigns them a cost proportional to the length of the reversed subsequence.
%assigns a cost to each operation based on the length of the interval in which it is performed.
This model has applications in comparative genomics, where it models a sequence of mutations in the evolution of a chromosome.
%The question, then, is about finding the worst-case \emph{cost} of sorting. 
Bender et al.~\cite{BGHHPSS04} showed that any input sequence can be sorted in this model with worst-case cost $O(n \log^2 n)$, and gave an $\Omega(n \log n)$ lower bound.
%Improved bounds for this model were given by Bender et al.~\cite{BGHHPSS04}. 
%Bender et al. proved improved bounds in this model~\cite{BGHHPSS04}.

%It was introduced by Pinter and Skiena~\cite{PS02} and later studied by Bender~\etal~\cite{BGHHPSS04} who showed how to sort any input sequence in this model with worst-case cost $O(n \log^2 n)$, and gave an $\Omega(n \log n)$ lower bound.

%In this section we show that the approach of Lemma~\ref{lemma:noncontiguous} gives a more general result about sorting using reversals of non-continguous sequences.

\begin{figure}[htb]
\centering
\includegraphics{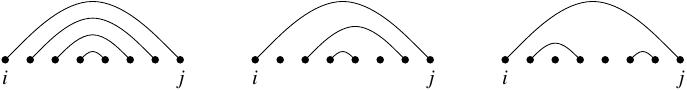}
\caption{An illustration of the swaps permitted in the contiguous reversals model (left), the non-contiguous reversals model (middle), and the non-crossing swaps model (right).}
\label{fig:sorting-models}
\end{figure}

We present a more general length-weighted sorting model, based on swaps.
%Consider the following general model of sorting.
Given a permutation $\rho$, pick an interval $[i..j]$ with $1 \leq i < j \leq n$ and $j - i = k$ and perform a set of swaps in it.
The cost of this operation is $O(k)$.
Adding constraints on the kinds of swaps allowed inside the interval leads to more specific models (see~Figure~\ref{fig:sorting-models}):

\begin{description}
 \item[Contiguous reversals] Reverse the entire interval $[i..j]$.
 \item[Non-contiguous reversals] Pick a subsequence $i_1, \ldots , i_m$ of $[i..j]$ and reverse its elements, without changing the position of any element outside the subsequence.
 \item[Non-crossing swaps] Swap a set $\{(i_1, j_1), \ldots , (i_m, j_m)\}$ of pairs inside $[i..j]$ such that there do not exist $x, y \in [1..m]$ with $x < y$ for which $i_x \leq i_y \leq j_x \leq j_y$, i.e, either the two intervals $[i_x, j_x]$ and $[i_y, j_y]$ nest or they are disjoint.
\end{description}

\prc{To the best of our knowledge, the latter two models are novel.}
Of the three models presented, the non-crossing swaps model is the most general, as it can simulate both other models.
In fact, the non-crossing swaps model is more general than (length-weighted versions of) a large number of sorting models~\cite[Chapter 3]{guillaume2009combinatorics}. 
We use our results on flips to give an $\Omega(n \log n)$ lower bound on the worst-case cost of sorting in the non-crossing swaps model and therefore in all of the sorting models that can be simulated by it.
\prc{Note that our proof of Theorem~\ref{thm-upper-bound} can be interpreted as an $O(n \log n)$ upper bound on the cost of sorting in the non-contiguous reversals model, and thereby also in the non-crossing swaps model.}
We first show how to simulate a non-crossing swap with flips.
The proof is very similar to the proof of Lemma~\ref{lemma:noncontiguous}.

\begin{lemma}
Let $\rho$ and $\rho'$ be two permutations of $[1..n]$ such that $\rho'$ can be obtained from $\rho$ using one set of non-crossing swaps in an interval of length $k$. Then the canonical triangulation $(T^*, \rho)$ can be flipped to $(T^*, \rho')$ with $O(k)$ flips.
\end{lemma}
\begin{proof}
Let $S = [i..j]$ be the sequence of edges in which we perform the swaps. We proceed by induction on $|S|$. In the base case, $|S| = 0$ and we are done, so let $|S| = k$ and assume that the lemma holds for all $|S| < k$. If there is an edge $e$ in $S$ that is not involved in any swap, we flip it. This results in a smaller convex polygon and reduces $|S|$ by one, so we can perform the swaps there by induction and flip $e$ back afterwards.

If every edge is involved in a swap, consider a swap $x = (i_x, j_x)$ such that no other swap is nested inside $x$. The edges $i_x$ and $j_x$ must be consecutive, since all edges are involved in some swap, and swaps do not cross. Thus, they form a convex pentagon and can be swapped with five flips, as shown in Figure~\ref{fig:pentagonswap}. But instead of completing this swap sequence, we halt it after the first three flips. This leaves the edges $i_x$ and $j_x$ out of the way of the remaining edges, so we can perform the rest of the swaps by induction. Afterwards, we perform the final two flips to return $i_x$ and $j_x$. Thus, we can perform each swap with 5 flips, giving a total cost of $O(k)$.
\end{proof}

Combining this lemma with Theorem~\ref{thm:el-convex-lb} gives the desired lower bound.

\begin{theorem}
\label{thm:noncrossing}
The worst-case cost of sorting a permutation of $[1..n]$ in the non-crossing swaps model is $\Omega(n\log n)$.
\end{theorem}

\section{Spiral polygons}
\label{sec-spiral}

In this section, we prove the Orbit Conjecture  for spiral polygons, even when there are multiple orbits. We provide a complete characterization of the orbits, and prove a tight quadratic bound on the diameter of each connected component of the flip graph. Finally, we show how to test for the condition of the Orbit Conjecture in linear time.

\prc{A vertex of a polygon is a \emph{convex} vertex if the interior angle at the vertex is less than $\pi$, otherwise it is \emph{reflex}. A \emph{spiral polygon} is a polygon in which all reflex vertices are consecutive along the boundary.} For a spiral polygon $P$ with $m$ convex and $k$ reflex vertices, let $C = c_1c_2\ldots c_m$ denote the convex chain \Abluenote{in clockwise order} and $R = r_1r_2\ldots r_k$ denote the reflex chain \Abluenote{in counterclockwise order} such that $c_1$ and $r_1$ are adjacent, and $c_m$ and $r_k$ are adjacent. For any vertex $c$ on the convex chain (resp. reflex chain), let $V(c)$ be the set of vertices on the reflex chain (resp. convex chain) that are visible from $c$.

\begin{figure}[htb]
 \centering
 \begin{subfigure}[b]{0.48\textwidth}
  \centering
  \includegraphics{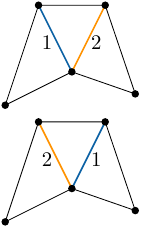}
  \caption{}
  \label{fig:spiral-disconnected}
 \end{subfigure}
 \begin{subfigure}[b]{0.48\textwidth}
  \centering
  \includegraphics{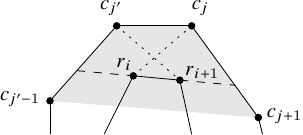}
  \caption{}
  \label{fig:proof-large-intersection}
 \end{subfigure}
 \caption{(a) Two edge-labelled triangulations of the same spiral polygon that cannot be transformed into each other by flips. (b) If \Abluenote{$V(r_i) \cap V(r_{i+1}) = \{c_{j}, c_{j'} \}$ then \prc{the} convex quadrilateral $c_{j-1}, c_{j}, c_{j'}, c_{j'+1}$} is not empty.}
\end{figure}

Consider a triangulated 5-vertex polygon with a single reflex vertex (see~Figure~\ref{fig:spiral-disconnected}). The two diagonals cannot be swapped with flips, thereby implying that the flip graph is disconnected. 
\Abluenote{The same thing happens with any spiral polygon with only 3 convex vertices---there is a unique triangulation and the flip graph is either a singleton (if there is at most one reflex vertex), or disconnected.
We now characterize when the flip graph is connected.
We only consider spiral polygons with at least 4 convex vertices.
A spiral polygon is \emph{locally convex} if it has at least four convex vertices and for every four consecutive vertices $a$, $b$, $c$, and $d$ on the convex chain, the %convex 
quadrilateral $abcd$ is empty.
Note that it is possible for four consecutive convex vertices to form a non-convex quadrilateral, but then the quadrilateral cannot be empty.
} 
Local convexity is the key to the characterization of the orbits and connectivity of the flip graph.

\begin{theorem}
\label{theorem-spiral-characterization}
For \Abluenote{an $n$-vertex spiral polygon $\mc{P}$ with at least 4 convex vertices}
%$n > 4$ vertices, 
the flip graph, $FG_{\mc{P}}$, is connected if and only if $\mc{P}$ is locally convex.
Furthermore, the diameter of each connected component of $FG_{\mc{P}}$ is $O(n^2)$.
% and there are cases when $FG_{\mc{P}}$ has a diameter of $\Omega(n^2)$.
\end{theorem}

We start by showing that  if $\mc{P}$ is locally convex, then its flip graph is connected. Observe the following:
\begin{lemma}
\label{lemma-locally-convex}
 If $\mc{P}$ is locally convex, then for any two consecutive vertices $r_i, r_{i+1}\in R$ on the reflex chain, $|V(r_i)\cap V(r_{i+1})| \geq 3$.
\end{lemma}
\begin{proof}
\Abluenote{Extend the segment $r_i, r_{i+1}$ in both directions until it hits the convex chain and exits the polygon.   Let $c_j, \ldots, c_{j'}$ be the vertices of the convex chain between these two exit points, noting that $j>1$ and $j' < m$.  Observe that $V(r_i)\cap V(r_{i+1}) = \{c_j, \ldots, c_{j'} \}$.  

Suppose that $|V(r_i)\cap V(r_{i+1})| \leq 2$.  We will show that $\mc{P}$ is not locally convex.  Consider the vertices $c_{j-1}, c_j, \ldots, c_{j'}, c_{j' +1}$.  They form a quadrilateral or triangle that contains $r_i$ and $r_{i+1}$.  If they form a quadrilateral, this violates local convexity, and if they form a triangle then we can extend it to a quadrilateral that violates local convexity.
}
\remove{
% note that this was not a proof by contradiction -- just using the contrapositive
Suppose, for contradiction, that $|V(r_i)\cap V(r_{i+1})| \leq 2$. Let $c_j$ be the rightmost vertex of $V(r_i)$ and $c_{j'}$ be the leftmost vertex of $V(r_{i+1})$ (See Figure~\ref{fig:proof-large-intersection}). Since any polygon can be triangulated, there must be a convex vertex visible to both $r_i$ and $r_{i+1}$, so we have that either $j = j'$ or $j'=j-1$. Consider the vertices $c_{j'-1}$ and $c_{j+1}$. These must exist, since $c_1$ and $c_m$ are visible only from reflex vertices $r_1$ and $r_k$, respectively. Now consider the infinite ray starting at $r_{i+1}$ in the direction of $r_{i+1}c_{j'}$ and rotate it towards $c_{j'-1}$ until it hits a point inside the triangle $r_{i+1}c_{j'}c_{j'-1}$. Since $r_{i+1}$ does not see $c_{j'-1}$, this ray must hit $r_i$ before it hits $c_{j'-1}$. That means $r_i$ lies inside the triangle $r_{i+1}c_{j'}c_{j'-1}$. Similarly, $r_{i+1}$ lies inside the triangle $r_ic_jc_{j+1}$. This is possible only if both $r_i$ and $r_{i+1}$ lie inside the quadrilateral $c_{j'-1}c_{j'}c_jc_{j+1}$ (which is a triangle in case $j=j'$). But that would mean that the polygon was not locally convex.
}
\end{proof}

To prove connectivity of the flip graph, we define the following canonical triangulation $\T'$,
together with a partition of its diagonals into {\it fans}. 
The canonical triangulation will be defined by specifying the triangle incident to each edge of the polygon.
%Let each diagonal have a unique label from the set $\{1, \ldots, n\}$.
Let $r_ir_{i+1}$ be an edge of the reflex chain with $V(r_i)\cap V(r_{i+1}) = \{c_j,\ldots ,c_{j'}\}$.  The apex of the triangle on edge $r_ir_{i+1}$ is defined to be $c_{\left\lfloor (j+j')/2\right\rfloor}$.  
The diagonals of all such triangles incident to a given convex vertex $c$ constitute the \emph{convex fan} at $c$.  
See Figure~\ref{fig:convex-fan}.   
After this, each edge $c_jc_{j+1}$ has a unique reflex vertex that can be its apex.  We add these triangles to $\T'$.  The new diagonals incident to reflex vertex $r$ constitute the \emph{reflex fan} at $r$. 
See Figure~\ref{fig:reflex-fan}.
Note that every diagonal of $\T'$ connects a convex and a reflex vertex, and thus there is a natural ordering of the diagonals of $\T'$ \Abluenote{consistent with the ordering of the convex and reflex chains}.  We label the diagonals of $\T'$ with labels $\{1, \ldots, |E|\}$ in this order ($|E|$ is the number of diagonals in $\T'$).
The fans partition the diagonals of $\T'$ and there is a natural ordering of the fans 
$D_1, \ldots, D_t$ such that the diagonals of $D_i$ occur before the diagonals of $D_j$ if $i < j$. 
%To complete the triangulation, connect each edge $c_jc_{j+1}$ with the unique reflex vertex that can be its apex. Notice that this triangulation can be decomposed into an ordered sequence of maximal convex and reflex fans. A convex fan is a convex vertex adjacent to a contiguous sequence of reflex vertices and a reflex fan is a reflex vertex adjacent adjacent to a contiguous sequence of convex vertices. The ordering of the fans induces an ordering on the diagonals. Let the ordering on the edge labels be the sorted order. 
%
%Finally, we 
We now provide a flip sequence to move from a given edge-labelled triangulation $\T$ to the canonical triangulation $\T'$ using $O(n^2)$ flips. %To do so, it suffices to show how to move from $\T$ to the canonical triangulation.
As flips are reversible, this implies that we can transform any edge-labelled triangulation into any other.

\begin{figure}[htb]
 \centering
 \begin{subfigure}[b]{0.4\textwidth}
  \centering
  \includegraphics{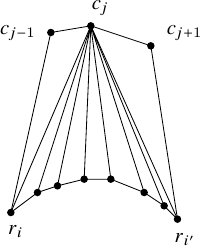}
  \caption{}
  \label{fig:convex-fan}
 \end{subfigure}
 \begin{subfigure}[b]{0.58\textwidth}
  \centering
  \includegraphics{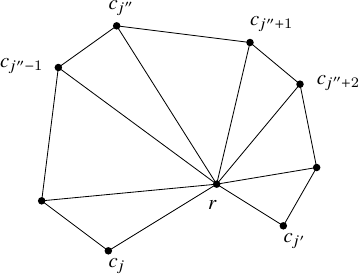}
  \caption{}
  \label{fig:reflex-fan}
 \end{subfigure}
 \caption{(a) A convex fan ($c_{j-1}$ and $c_{j+1}$ are not part of the fan). (b) A reflex fan.}
\end{figure}

We start by ignoring labels and transforming $\T$ into the unlabelled canonical triangulation with $O(n)$ flips~\cite{hanke1994ebene}. What remains is to sort the edge labels. Our flip sequence achieves this with $O(n^2)$ flips via insertion sort. 
If $d_1, \ldots, d_{|E|}$ is the ordered list of diagonals, the {\it insertion} of $d_j$ at $d_i$ for $i <j$ replaces the subsequence $d_i, \ldots, d_j$ by $d_j, d_i, d_{i+1}, \ldots, d_{j-1}$.
It suffices to show how to perform a single insertion using $O(n)$ flips. Based on the structure of the canonical triangulation, it is sufficient to show how to perform an insertion within 
one fan and between two consecutive fans
\Abluenote{with a number of flips proportional to the sizes of the fans}.

\begin{lemma}
\label{lemma-insertion-fan}
For each fan $D_i$, insertion between two edges of $D_i$ can be performed using $O(|D_i|)$ flips. In addition, swapping the last edge of $D_i$ with the first edge of $D_{i+1}$ can be performed using  $O(|D_i|+|D_{i+1}|)$ flips.
\end{lemma}
\begin{proof}
We first show how to perform an insertion in a convex fan $D_i$ (see Figure~\ref{fig:convex-fan}). Let $c_j$ be the common convex vertex and 
\Abluenote{$r_i,\ldots,r_{i'}$} be the sequence of reflex vertices. 
%$r_a,\ldots,r_{a+b}$ be the sequence of reflex vertices.   -- and notation changed for the rest of this paragraph and the next
This means that $c_j$ was chosen to be the apex vertex for each edge of the reflex chain between $r_i$ and $r_{i'}$ during the construction of the canonical triangulation. Using Lemma~\ref{lemma-locally-convex}, together with the fact that the apex is always chosen to be in the ``middle'' of the set of visible vertices, we get that $c_{j-1}$ and $c_{j+1}$ must exist and be visible to all reflex vertices of $D_i$.  
\Abluenote{Also note that the canonical triangulation contains diagonals $r_i c_{j-1}$ and $r_{i'} c_{j+1}$.}

We show how to insert diagonal \Abluenote{$c_jr_k$ to $c_jr_\ell$ for any $i \le \ell < k \le i'$}. First, flip all diagonals intersecting segment $r_{k}c_{j+1}$ and all diagonals intersecting $r_{k-1}c_{j-1}$. This results in the convex pentagon $r_{k-1}r_{k}c_{j+1}c_jc_{j-1}.$ Swap the two diagonals of this pentagon with five flips as shown in Figure~\ref{fig:pentagonswap} in Section~\ref{sec:convex}. Notice that this has moved $c_jr_{k}$ to $c_jr_{k-1}$. To move it to $c_jr_{k-2}$, we flip a constant number of edges to get the convex pentagon $r_{k-2}r_{k-1}c_{j+1}c_jc_{j-1}$. We continue in this way until $c_jr_{k}$ has reached $c_jr_{\ell}$. At that point, we flip all diagonals adjacent to $c_{j-1}$ and $c_{j+1}$ back to their original positions. Notice that the order of these diagonals is preserved and this results in the insertion of $c_jr_{k}$ to $c_jr_{\ell}$, using a total of $O(|D_i|)$ flips.

If $D_i$ is a reflex fan (Figure~\ref{fig:reflex-fan}), let $r$ be its common reflex vertex and 
\Abluenote{let $c_j , \ldots, c_{j'}$ be its}
%$C_{j,j'}$ be the set of 
convex vertices.
Consider two consecutive diagonals  $rc_{j''}$ and $rc_{j''+1}$ of $D_i$.  These form a triangle $t$ with apex $r$.
Consider the triangle before $t$.  If it had  
apex $c_{j''}$ then the diagonal $rc_{j''}$ would have been placed in the convex fan $D_{i-1}$.  Therefore it has apex $r$ and includes the diagonal (or boundary edge) $r c_{j''-1}$.  Similarly, the triangle after $t$ must have apex $r$ as well and includes the diagonal (or boundary edge) $r c_{j''+2}$.  
Now consider the pentagon formed by $r$ and the chain $c_{j''-1}\ldots c_{j''+2}$. Since $\mc{P}$ is locally convex, this pentagon must be empty and convex. Thus we can swap $rc_{j''}$ and $rc_{j''+1}$ with five flips as before. 
This shows that any two consecutive diagonals of a reflex fan can be swapped using five flips. Using a sequence of such swaps, we can perform any insertion with $O(|D_i|)$ flips.

Finally, we show how to swap %an edge between two fans.
the last edge $d$ of fan $D_i$ with the first edge $d'$ of fan $D_{i+1}$.
We consider four cases. The case when $D_i$ and $D_{i+1}$ are both convex fans is shown in Figure~\ref{figure-adjacent-insertion-a} and the case when $D_i$ is convex and $D_{i+1}$ is reflex is shown in Figure~\ref{figure-adjacent-insertion-b}. In both cases, let $d$ and $d'$ share vertex $r$. Our strategy is to perform a linear number of flips in the convex fans to get an empty pentagon 
%$efghr$ 
where $d$ and $d'$ can be swapped. 
\Abluenote{The pentagon consists of $r$ together with 4 consecutive vertices $efgh$ on the convex chain that are visible from $r$.  Local convexity implies that the pentagon is empty and convex.}

The case when $D_{i+1}$ is convex and $D_i$ is reflex is symmetric.
% to the case when $D_i$ is convex and $D_{i+1}$ reflex. 
Finally, it is not possible to have both $D_i$ and $D_{i+1}$ reflex because a reflex fan was defined to contain all the diagonals between two convex fans.
% \remove{Suppose, for contradiction, that they are both reflex and let $d$ be the rightmost diagonal of $D_i$ and $d'$ be the leftmost diagonal of $D_{i+1}$. Clearly, $d$ and $d'$ share an endpoint, which cannot be a reflex vertex, otherwise all diagonals of $D_i\cup D_{i+1}$ would share the same endpoint and they would not be in different sets. If $d$ and $d'$ share a convex vertex $c$, then that convex vertex is incident to at least two diagonals and hence our procedure for constructing $T$ would have put them both into the convex fan corresponding to $c$.
% }
\end{proof}

\begin{figure}[htb]
 \centering
 \begin{subfigure}[b]{\textwidth}
  \centering
  \includegraphics{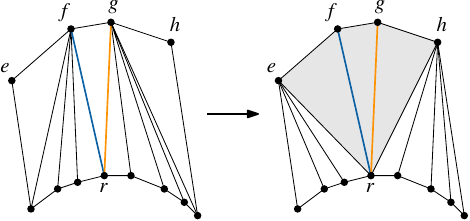}
  \caption{}
  \label{figure-adjacent-insertion-a}
 \end{subfigure}
 \begin{subfigure}[b]{\textwidth}
  \centering
  \includegraphics{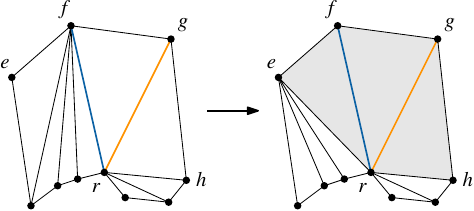}
  \caption{}
  \label{figure-adjacent-insertion-b}
 \end{subfigure}
 \caption{Swapping edges between adjacent fans $D_i$ and $D_{i+1}$ using the pentagon $efghr$.
 (a) $D_i$ and $D_{i+1}$ are both convex fans. (b) $D_i$ is a convex fan but $D_{i+1}$ is a reflex fan.}
\end{figure}

Using Lemma~\ref{lemma-insertion-fan}, we can reorder the labels of the edges in a locally convex spiral polygon with $O(n^2)$ flips by simulating insertion sort. What remains to be shown is that if the flip graph is connected then $\mc{P}$ is locally convex. We show this in the contrapositive.

%The following lemma is useful.
%\begin{lemma}
%\label{lemma-non-empty-certificate}
%Let $c_j, c_{j+1}, c_{j+2}, c_{j+3}$ be four consecutive vertices on $C$ such that the quadrilateral $c_jc_{j+1}c_{j+2}c_{j+3}$ is not empty. Let $r_i$ be the rightmost vertex of $R$ that is visible from $c_j$. Then $r_i$ must lie inside the quadrilateral $c_jc_{j+1}c_{j+2}c_{j+3}$. 
%\end{lemma}
%\begin{proof}
%Consider the ray starting at $c_j$ and extending in the direction of $c_jc_{j+1}$ and rotate it clockwise around $c_j$ (Figure~\ref{figure-non-empty-certificate}(a)). Let $r_i$ be the first vertex inside the quadrilateral $c_jc_{j+1}c_{j+2}c_{j+3}$ that it hits. If it hits more than one vertex of $R$ simultaneously, then pick the leftmost one to be $r_i$. It is clear that $r_i$ can see $c_j$ but $r_{i+1}$ cannot. Thus $r_i$ must be the rightmost vertex of $R$ visible to $c_j$. 
%\end{proof}

\begin{figure}[htb]
 \centering
 \begin{subfigure}[b]{0.48\textwidth}
  \centering
  \includegraphics{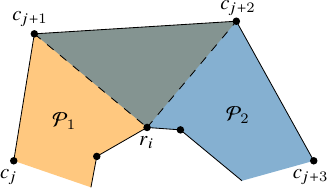}
  \caption{}
  \label{figure-non-empty-certificate-a}
 \end{subfigure}
 \begin{subfigure}[b]{0.48\textwidth}
  \centering
  \includegraphics{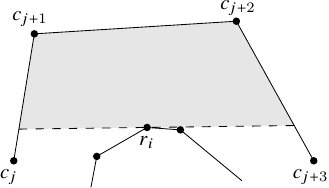}
  \caption{}
  \label{figure-non-empty-certificate-b}
 \end{subfigure}
 \caption{(a) Decomposition of a spiral polygon that is not locally convex into two sub-polygons. 
(b) To find $r_i$, we move a segment from $c_{j+1}c_{j+2}$ to $c_{j} c_{j+3}$, stopping when we hit the reflex chain.}
\end{figure}

% Finally, we show exactly how $QG$ is disconnected if $P$ is not locally convex.

\begin{lemma}
\label{lemma-decomposition}
If there exist four consecutive vertices $c_j$, $c_{j+1}$, $c_{j+2}$, $c_{j+3}$ on $C$ that form a non-empty quadrilateral, then there exist two spiral polygons $\mc{P}_1$ and $\mc{P}_2$, such that: (a) each boundary edge of $\mc{P}_1$ and $\mc{P}_2$ is either a boundary edge or a diagonal of $\mc{P}$, (b) each diagonal of $\mc{P}$ is a diagonal of either $\mc{P}_1$ or $\mc{P}_2$, and (c) no diagonal of $\mc{P}_1$ can flip to a diagonal of $\mc{P}_2$.
\end{lemma}
\begin{proof}
\Anote{
We claim that there is a vertex $r_i$ on $R$ that sees both $c_{j+1}$ and $c_{j+2}$, and such that $c_j$ does not see $r_{i+1}$ and $c_{j+3}$ does not see $r_{i-1}$.  (Justification below.)
Given such a vertex $r_i$, let $\mc{P}_1$ and $\mc{P}_2$ be the spiral polygons defined by the chains $c_1, \ldots, c_{j+2}$ and $r_1, \ldots, r_i$, and $c_{j+1}, \ldots, c_{m}$  and $r_i, \ldots, r_k$, respectively.  
\Abluenote{Then $\mc{P}_1$ and $\mc{P}_2$ overlap on triangle $r_i c_{j+1} c_{j+2}$.} See Figure~\ref{figure-non-empty-certificate-a}.

% Let $r_x$ be a vertex on $R$ that sees both $c_{j+1}$ and $c_{j+2}$, such that $c_j$ does not see $r_{x+1}$ and $c_{j+3}$ does not see $r_{x-1}$. Let $\mc{P}_1$ and $\mc{P}_2$ be the spiral polygons defined by the chains $c_1, \ldots, c_{j+2}$ and $r_1, \ldots, r_x$, and $c_{j+1}, \ldots, c_{m}$  and $r_x, \ldots, r_k$, respectively.
%We first show that $\mc{P}_1$ and $\mc{P}_2$ meet all conditions.
 
It is easy to see that $\mc{P}_1$ and $\mc{P}_2$ satisfy condition (a), as $c_{j+1}r_i$ and $c_{j+2}r_i$ are diagonals of $\mc{P}$. For condition (b) and (c), we note the following. By construction, 
there is no diagonal between a vertex of $\mc{P}_1 - \mc{P}_2$ and a vertex of $\mc{P}_2 - \mc{P}_1$.
%no diagonal exists between a vertex on the convex chain of $\mc{P}_1$ (resp. $\mc{P}_2$) and reflex chain of $\mc{P}_2$ (resp. $\mc{P}_1$). 
This implies that the diagonals of $\mc{P}$ are partitioned between $\mc{P}_1$ and $\mc{P}_2$. For condition (c), suppose a diagonal of $\mc{P}_1$ flipped to one in $\mc{P}_2$. Then they must intersect, 
so they must be incident to $c_{j+2}$ and $c_{j+1}$ respectively.  But then they cannot form an empty convex quadrilateral.
}
%Since both diagonals must have endpoints on the convex chain, the quadrilateral formed by their endpoints cannot be empty.

To find $r_i$, imagine the continuous motion of a segment whose initial position is  $c_{j+1}c_{j+2}$ and whose final position is $c_jc_{j+3}$ \prc{(see Figure~\ref{figure-non-empty-certificate-b})}.  One endpoint of the segment moves at uniform speed in a straight line from $c_{j+1}$ to $c_j$, and the other endpoint moves at uniform speed in a straight line from $c_{j+2}$ to $c_{j+3}$.  The initial segment does not intersect the reflex chain; the final segment does.  Therefore, by continuity, some intermediate segment, say $pq$, lies inside the polygon but contains a reflex vertex (or possibly two).  We choose one such vertex to be $r_i$.  The area swept by the segment until reaching $pq$ forms an empty convex quadrilateral with $c_{j+1}$, $c_{j+2}$, and $r_i$ on the boundary. Thus, $r_i$ sees both $c_{j+1}$ and $c_{j+2}$. Furthermore, segment $pq$ is a chord of $\mc{P}$ that goes through a reflex vertex, namely $r_i$. It is well known\footnote{See, for instance, Bose~\etal~\cite{bose2002efficient}, Lemma 1.} that such a chord partitions $\mc{P}$ into three subpolygons $\mc{P}_A$, $\mc{P}_B$, and $\mc{P}_C$ such that no two vertices in the non-adjacent subpolygons $\mc{P}_A$ and $\mc{P}_C$ can see each other. As $c_j \in \mc{P}_A$ and $r_{i+1} \in \mc{P}_C$, this implies that $c_j$ cannot see $r_{i+1}$. Analogously, $c_{j+3}$ cannot see $r_{i-1}$, showing that $r_i$ has all the properties we were looking for.
\remove{
% Sander's construction for r_x 
To find $r_x$, we can take the line $\ell$ through the edge $c_{j+1}c_{j+2}$ and sweep it towards the interior of the quadrilateral, until we hit a vertex in, or on the boundary of the quadrilateral. 
See Fig.~\ref{figure-non-empty-certificate}(b,c).
If we hit a vertex in the interior, this is $r_x$. Otherwise, we hit either $c_j$ or $c_{j+3}$. In this case, we stop the sweep and instead rotate $\ell$ towards $c_jc_{j+3}$ until we hit a vertex in the interior of the quadrilateral, which is $r_x$. We are guaranteed to hit a vertex before reaching $c_jc_{j+3}$, because the quadrilateral is non-empty. In either case, the area swept by $\ell$ forms an empty convex quadrilateral with $c_{j+1}$, $c_{j+2}$, and $r_x$ on the boundary. Thus, $r_x$ sees both $c_{j+1}$ and $c_{j+2}$. Furthermore, the part of $\ell$ inside the quadrilateral is a chord of $\mc{P}$ that goes through a reflex vertex, namely $r_x$. It is well known\footnote{See, for instance, Bose~\etal~\cite{bose2002efficient}, Lemma 1.} that such a chord partitions $\mc{P}$ into three subpolygons $\mc{P}_A$, $\mc{P}_B$, and $\mc{P}_C$ such that no two vertices in the non-adjacent subpolygons $\mc{P}_A$ and $\mc{P}_C$ can see each other. As $c_j \in \mc{P}_A$ and $r_{x+1} \in \mc{P}_C$, this implies that $c_j$ cannot see $r_{x+1}$. Analogously, $c_{j+3}$ cannot see $r_{x-1}$, showing that $r_x$ has all the properties we were looking for.
}
%%%% Original full proof %%%%
\remove{
Let $r_x$ be a vertex on $R$ that sees both $c_{j+1}$ and $c_{j+2}$, such that $c_j$ does not see $r_{x+1}$ and $c_{j+3}$ does not see $r_{x-1}$. Let $\mc{P}_1$ and $\mc{P}_2$ be the spiral polygons defined by the chains $C_{1, j+2}$ and $R_{1, x}$, and $C_{j+1, m}$ and $R_{x, k}$, respectively. We first show that $\mc{P}_1$ and $\mc{P}_2$ meet all conditions, then we show that we can always find a vertex $r_x$.
 
 It is easy to see that $\mc{P}_1$ and $\mc{P}_2$ satisfy condition (a), as $c_{j+1}r_x$ and $c_{j+2}r_x$ are diagonals of $\mc{P}$. Condition (b) has two parts. First, there are no diagonals shared between $\mc{P}_1$ and $\mc{P}_2$, because the diagonals defined by the shared vertices ($c_{j+1}$, $c_{j+2}$, and $r_x$) are all on the boundary of either $\mc{P}_1$ or $\mc{P}_2$. Second, since $c_j$ does not see $r_{x+1}$ or $c_{j+3}$ and visibility is monotonic, no diagonal can connect a vertex on $C_{1, j}$ to either $C_{j+3, m}$ or $R_{x+1, k}$. Likewise, since $c_{j+3}$ does not see $r_{x-1}$, no diagonal can connect $C_{j+3,m}$ to $R_{1, x-1}$. Thus, the two endpoints of each diagonal of $\mc{P}$ must be in the same polygon. Finally, to prove that $\mc{P}_1$ and $\mc{P}_2$ meet condition (c), consider two diagonals, $e_1$ of $\mc{P}_1$ and $e_2$ of $\mc{P}_2$, such that $e_1$ flips to $e_2$. Then $e_1$ and $e_2$ must intersect properly. Since $e_1$ lies in the interior of $\mc{P}_1$ and $e_2$ lies in the interior of $\mc{P}_2$, their intersection point must lie in $\mc{P}_1 \cup \mc{P}_2$: the triangle $X = c_{j+1}r_xc_{j+2}$. As $X$ forms an ear in both $\mc{P}_1$ and $\mc{P}_2$, this means that $c_{j+2}$ is an endpoint of $e_1$, and $c_{j+1}$ is an endpoint of $e_2$. Note that $r_x$ cannot be an endpoint of either edge, as $c_{j+1}r_x$ and $c_{j+2}r_x$ do not lie in the interior of $X$. Let $a \in \mc{P}_1$ and $b \in \mc{P}_2$ be the other endpoints of $e_1$ and $e_2$. Then $ab$ must be a diagonal or boundary edge of $\mc{P}$. But for condition (b) we showed that such a diagonal cannot exist, and the boundary vertices of $\mc{P}_1$ and $\mc{P}_2$ are separated by $r_x$. Thus, no diagonal of $\mc{P}_1$ flips to a diagonal of $\mc{P}_2$.

 To find $r_x$, we take the line $\ell$ through the edge $c_{j+1}c_{j+2}$ and sweep it towards the interior of the quadrilateral, until we hit a vertex in, or on the boundary of the quadrilateral. If we hit a vertex in the interior, this is $r_x$. Otherwise, we hit either $c_j$ or $c_{j+3}$. In this case, we stop the sweep and instead rotate $\ell$ towards $c_jc_{j+3}$ until we hit a vertex in the interior of the quadrilateral, which is $r_x$. We are guaranteed to hit a vertex before reaching $c_jc_{j+3}$, because the quadrilateral is non-empty. In either case, the area swept by $\ell$ forms an empty convex quadrilateral with $c_{j+1}$, $c_{j+2}$, and $r_x$ on the boundary. Thus, $r_x$ sees both $c_{j+1}$ and $c_{j+2}$. Furthermore, the part of $\ell$ inside the quadrilateral is a chord of $\mc{P}$ that goes through a reflex vertex -- $r_x$. It is well known\footnote{See, for instance, Bose~\etal~\cite{bose2002efficient}, Lemma 1.} that such a chord partitions $\mc{P}$ into three subpolygons $A$, $B$, and $C$ such that no two vertices in the non-adjacent subpolygons $A$ and $C$ can see each other. As $c_j \in A$ and $r_{x+1} \in C$, this implies that $c_j$ cannot see $r_{x+1}$. Analogously, $c_{j+3}$ cannot see $r_{x-1}$, showing that $r_x$ has all the properties we were looking for.
}
\end{proof}

%Let $\mc{P}_1$ be the spiral polygon defined by the chains $C_{1, j+2}$ and $R_{1, i}$ (Figure~\ref{figure-non-empty-certificate}(b,c)). The definition of $\mc{P}_2$ and $\mc{P}_3$ depends on whether $c_{j+1}$ is visible from $r_{i+1}$. If it is, let $\mc{P}_3$ be the polygon defined by $C_{j+1, m}$ and $R_{i+1, k}$ and let $\mc{P}_2$ be the quadrilateral $c_{j+1}c_{j+2}r_{i+1}r_i$; otherwise let $\mc{P}_3$ be the polygon defined by $C_{j+1, m}$ and $R_{i, k}$ and let $\mc{P}_2$ be the empty polygon.
%
%With these definitions, it is easy to check that conditions (1) and (2) above are satisfied. For (3), suppose, for contradiction, that there are diagonals $d_s$ of $\mc{P}_s$ and $d_t$ of $\mc{P}_t$ for $s\neq t$ such that they are connected by an edge in $QG$. This means $d_s$ and $d_t$ must cross, which is possible only if $d_s$ is incident on $c_{j+1}$ and $d_t$ is incident on $c_{j+2}$. However, with this constraint, whatever we pick for the other endpoints of $d_s$ and $d_t$, the convex quadrilateral formed by their four endpoints will always contain $r_i$.

This concludes the proof of Theorem~\ref{theorem-spiral-characterization}. It also gives us a way to decide, in polynomial time, whether a given spiral polygon has a connected flip graph: for each quadrilateral formed by four consecutive vertices of the convex chain, check if it is empty. A naive implementation of this procedure will take $O(n^2)$ time.
%: there are $O(n)$ quadrilaterals to check and for each quadrilateral, there are at most $n$ points that could potentially lie inside it.   %deleted for space saving AL
However, we can reduce the running time to $O(n)$  
by walking simultaneously along the convex and reflex chains of the polygon and using the fact that visibility is monotonic.
%using Lemma~\ref{lemma-decomposition}. 
Thus we get the following theorem.

\begin{theorem}
\label{theorem-linear-time}
Given a spiral polygon $\mc{P}$ with $n$ vertices, we can decide in linear time whether the flip graph of the edge-labelled triangulations of $\mc{P}$ is connected.
\end{theorem}

%%%%%%%%%%%%%%%%%%%%%%%%%%%%%%%%%%%%%
%\subsection{Connectivity between two given triangulations}
%We start with the following lemma whose proof can be found in Appendix~\ref{appendix-full-decomposition}.
%\begin{lemma}
%\label{lemma-full-decomposition}
%Given any spiral polygon $\mc{P}$, there exists a set $\{\mc{P}_1\ddd\mc{P}_t\}$ of spiral polygons such that: (a) for each $\mc{P}_i$, each boundary edge of $\mc{P}_i$ is either a boundary edge or a diagonal of $\mc{P}$, (b) each diagonal of $\mc{P}$ is a diagonal of exactly one of $\mc{P}_1\ddd\mc{P}_t$, (c) no diagonal of $\mc{P}_i$ is connected by an edge in $QG$ to a diagonal of $\mc{P}_j$ for $i\neq j$, and (d) each $\mc{P}_i$ is locally convex.
%%\begin{enumerate}
%%\item For each $\mc{P}_i$, each boundary edge of $\mc{P}_i$ is either a boundary edge or a diagonal of $\mc{P}$.
%%\item Each diagonal of $\mc{P}$ is a diagonal of exactly one of $\mc{P}_1\ddd\mc{P}_t$.
%%\item No diagonal of $\mc{P}_i$ is connected by an edge in $QG$ to a diagonal of $\mc{P}_j$ for $i\neq j$.
%%\item Each $\mc{P}_i$ is locally convex.
%%\end{enumerate}
%\end{lemma}
%%\begin{proof}
%%If $\mc{P}$ is already locally convex, we are done. Otherwise $\mc{P}$ has four consecutive vertices on its convex chain that form a non-empty quadrilateral. Using Lemma~\ref{lemma-decomposition}, we get three polygons on which we recurse to obtain the desired set $\{\mc{P}_1\ddd\mc{P}_t\}$.
%%\end{proof}

Theorem \ref{theorem-spiral-characterization} gives a complete characterization of the orbits in a spiral polygon: The orbit of a diagonal $d$ in $\mc{P}$ is precisely the set of diagonals in the maximal locally convex subpolygon of $\mc{P}$ that contains $d$. 
With this characterization, we now prove the Orbit Conjecture for spiral polygons.

\begin{theorem}
Given a spiral polygon $\mc{P}$ with $n$ vertices and two of its edge-labelled triangulations $\T_1$ and $\T_2$, there exists a flip sequence that transforms $\T_1$ to $\T_2$ if and only if for each label $\lambda$, the diagonal of $\T_1$ with label $\lambda$ is in the same orbit as the diagonal of $\T_2$ with label $\lambda$. Moreover, we can decide in $O(n)$ time if such a flip sequence exists, and if it does, it is of length at most $O(n^2)$.
\end{theorem}
\begin{proof}
The `only if' direction is easy. For the `if' direction, we provide a flip sequence.
First ignore labels and transform both $\T_1$ and $\T_2$ into the unlabelled canonical triangulation $\T$. Then rearrange the labels in $\T$. If the diagonal with label $\lambda$ in $\T_1$ was in the same orbit as the diagonal with label $\lambda$ in $\T_2$, they must be diagonals of the same locally convex subpolygon $\mc{P}_i$. By applying Theorem \ref{theorem-spiral-characterization}, we can rearrange the labels inside each $\mc{P}_i$ with at most $O(|\mc{P}_i|^2)$ flips. We can find the decomposition into maximal locally convex subpolygons in linear time by walking simultaneously along the convex and reflex chains of the polygon, and exploiting the fact that visibility is monotonic. Once we have the decomposition, we can verify in constant time whether two edges with the same label are in the same orbit. 
\end{proof}

%It is easy to compute the maximal locally convex subpolygons in $O(n)$ time using a slight modification of the algorithm of Theorem~\ref{theorem-linear-time}. Using the decomposition, we can also check, in linear time, whether for each label $\lambda$, the diagonals of $\T_1$ and $\T_2$ with that label lie in the same polygon of the decomposition. This gives us the following.
%\begin{theorem}
%Given a spiral polygon $\mc{P}$ and two of its edge-labelled triangulations $\T_1$ and $\T_2$, one can decide in $O(n)$ time whether there exists a sequence of flips that transforms $\T_1$ to $\T_2$.
%\end{theorem}

\begin{figure}[htb]
\centering
\includegraphics{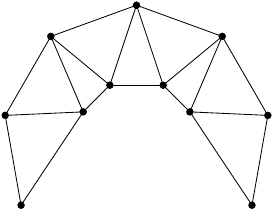}
\caption{A spiral polygon whose flip graph has diameter $\Theta(n^2)$.}
\label{fig:spiral-lowerbound}
\end{figure}

Finally, note that our quadratic bound on the flip distance is tight. Consider, for example, a locally convex spiral polygon for which every reflex vertex $r$ has $|V(r)|\leq 4$ (see Figure~\ref{fig:spiral-lowerbound}). Let $d_1\ddd d_{|E|}$ be the diagonals of the canonical triangulation of this polygon listed left-to-right. It is clear that the distance in the quadrilateral graph between any $d_i$ and $d_j$ is $\Omega(|j-i|)$. Thus, \prc{we can find $\Omega(n)$} pairs of diagonals at distance $\Omega(n)$ %$\Omega(n^2)$
in the quadrilateral graph. Since any flip sequence that moves a label from $d_i$ to $d_j$ must contain at least $\Omega(|j-i|)$ flips of that label, this gives a lower bound of $\Omega(n^2)$.

\begin{theorem}
There exists a spiral polygon with $n$ vertices that has a connected flip graph whose diameter is $\Omega(n^2)$.
\end{theorem}

\section{Lower bound on flips in edge-labelled polygons}
In this section, we provide an example of a polygon with two reflex chains where the diameter of the flip graph is $\Theta(n^3)$. The example also demonstrates a difficulty in generalizing our results on spiral polygons to more general polygons. 
Our example is based on a \emph{channel}~\cite{HNU99} that consists of two reflex chains visible to each other and joined by two edges.  
We define an \emph{augmented channel} to consist of two chains $A = a_1\ldots a_m$ and $B = b_1\ldots b_{m+2}$, as shown in Figure~\ref{figure-augmented-channel}, such that: (a) every vertex on $A$ is visible from every vertex on $B$, (b) the vertices on $A$ form a reflex chain, and (c) all vertices on $B$, except for $b_1$, $b_2$, $b_{m+1}$, and $b_{m+2}$ are reflex.
%We call our example an \emph{augmented channel}. It consists of two chains $A = a_1\ldots a_m$ and $B = b_1\ldots b_{m+2}$, as shown in Figure~\ref{figure-augmented-channel}, such that:
%\begin{itemize}
%\item[$\bullet$] Every vertex on $A$ is visible from every vertex on $B$.
%\item[$\bullet$] The vertices on $A$ form a reflex chain.
%\item[$\bullet$] All vertices on $B$, except for $b_1$, $b_2$, $b_{m+1}$, and $b_{m+2}$ are reflex.
%\end{itemize}

%The difference between an augmented channel and a channel is the presence of the two convex vertices $b_2$ and $b_{m+1}$ on $B$. 

\begin{figure}[htb]
 \centering
 \begin{subfigure}[b]{0.48\textwidth}
  \centering
  \includegraphics{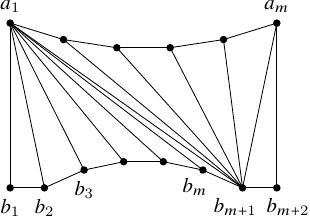}
  \caption{}
 \end{subfigure}
 \begin{subfigure}[b]{0.48\textwidth}
  \centering
  \includegraphics{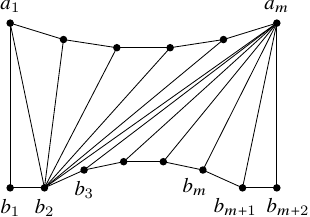}
  \caption{}
 \end{subfigure}
 \caption{(a) The left-inclined triangulation of an augmented channel. (b) The right-inclined triangulation of an augmented channel.}
 \label{figure-augmented-channel}
\end{figure}

\begin{theorem}
The diameter of the flip graph of an augmented channel is $\Theta(n^3)$.
\end{theorem}
\begin{proof}
We first show the upper bound. Any triangulation of the augmented channel induces an ordering on its diagonals as follows. If all diagonals connect a vertex on $A$ with a vertex on $B$, then the ordering induced is the same as the order in which a ray passing through the polygon from left to right will hit each diagonal. If, however, one or both of the diagonals $b_1b_3$ and $b_mb_{m+2}$ are present, then 
we take the ordering we would get by flipping those diagonals.
%(noting that the flipped diagonals connect $A$ and $B$).
%we assign to the diagonals the same ordering as we would in the triangulation obtained on flipping the ones that are present. 
The ordering on the diagonals defines an ordering on the labels.

Hurtado et al.~\cite[Theorem 3.8]{HNU99} showed that in the unlabelled setting, any triangulation of a channel can be transformed into any other triangulation of the channel using $O(n^2)$ flips.
We observe that their transformation preserves the ordering of the labels.
%
%Note also that given an ordering on the labels, all triangulations that induce that ordering can be transformed into each other with $O(n^2)$ flips (using Theorem 3.8 from Hurtado et al.~\cite{HNU99}). 
Thus for the $O(n^3)$ bound on the diameter of the flip graph, the challenge is in rearranging the labels. For any triangulation, let $D_1$ denote the set consisting of the first $m$ labels and $D_2$ denote the set consisting of the last $m$ labels. Since there are $2m-1$ labels in total, $D_1$ and $D_2$ will have one label in common, namely, the $m\textsuperscript{th}$ label.  
We will use two particular triangulations of the augmented channel: the \emph{left-inclined} and \emph{right-inclined} triangulations shown in Figure~\ref{figure-augmented-channel}. 
Note that in the left-inclined triangulation, %(Figure~\ref{figure-augmented-channel}(a)), 
$D_2$ is exactly the set of diagonals incident on $b_{m+1}$ and thus forms a convex fan. Similarly, in the right-inclined triangulation, % (Figure~\ref{figure-augmented-channel}(a)), 
$D_1$ forms a convex fan. This means both $D_1$ and $D_2$ can be sorted in $O(n^2)$ flips and any insertion in $D_1$ or $D_2$ can be performed in $O(n)$ flips using Lemma~\ref{lemma-insertion-fan}.

Thus our strategy is as follows. Given $\T_1$ and $\T_2$, first ignore labels and transform both into the right-inclined triangulation. Then, if $D_1$ contains a label bigger than $m$, insert it \prc{(using the spiral polygon formed by $b_1$, $b_2$, $b_3$, and chain $A$)} at $b_2a_m$ and transform the triangulation into a left-inclined triangulation while preserving the order of labels. Next, if $D_2$ contains a label smaller than $m$, insert the label at $a_1b_{m+1}$ and transform back into the right-inclined triangulation and repeat until $D_1$ only contains labels $1\ddd m$ and $D_2$ only contains labels $m\ddd 2m-1$. Finally, sort $D_1$ and then sort $D_2$. Since transforming between the left-inclined and right-inclined triangulations takes $O(n^2)$ flips and insertion inside $D_1$ or $D_2$ takes $O(n)$ flips, we get a bound of $O(n^3)$ flips.

To show the lower bound, we use the observation that to move a label from $D_2\setminus D_1$ to $D_1\setminus D_2$, we must first move that label from $D_2\setminus D_1$ to the $m^\textsuperscript{th}$ diagonal and then to $D_1\setminus D_2$. The first step can be performed only if the diagonal $a_1b_{m+1}$ is present and the second step can be carried out only if the diagonal $b_2a_m$ is present. Going from a triangulation that has the diagonal $a_1b_{m+1}$ to a triangulation that has the diagonal $b_2a_m$ requires $\Omega(n^2)$ flips using the argument from Hurtado et al.~\cite{HNU99}. Thus if we want to transform a triangulation with $D_1 = \{m\ddd 2m-1\}$ and $D_2=\{1\ddd m\}$ to a triangulation with $D_1 = \{1\ddd m\}$ and $D_2=\{m\ddd 2m-1\}$, we will need at least $\Omega(n^3)$ flips.
\end{proof}

The theorem above demonstrates that deciding whether, say, $a_1b_m$ and $a_1b_{m-1}$ are connected depends on the exact position of $b_2$ and $b_{m+1}$. In the example above they are connected, but if $b_2$ and $b_{m+1}$ were reflex vertices, they would not be.

\section{Combinatorial triangulations}

In this section, we prove the Orbit Conjecture for edge-labelled combinatorial triangulations and show that the diameter of the flip graph is $\Theta(n \log n)$. Note that we consider two edge-labelled triangulations to be equivalent if they have an isomorphism that preserves the edge labels (in other words, the vertices are unlabelled).
 
 \begin{figure}[htb]
  \centering
  \includegraphics{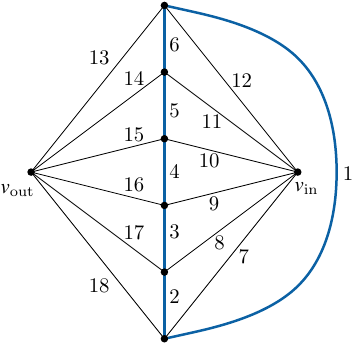}
  \caption{The labelled canonical triangulation on 8 vertices. The spine is indicated in bold.}
  \label{fig:comb-canonical}
 \end{figure}

\begin{figure}[p]
   \centering
   \includegraphics{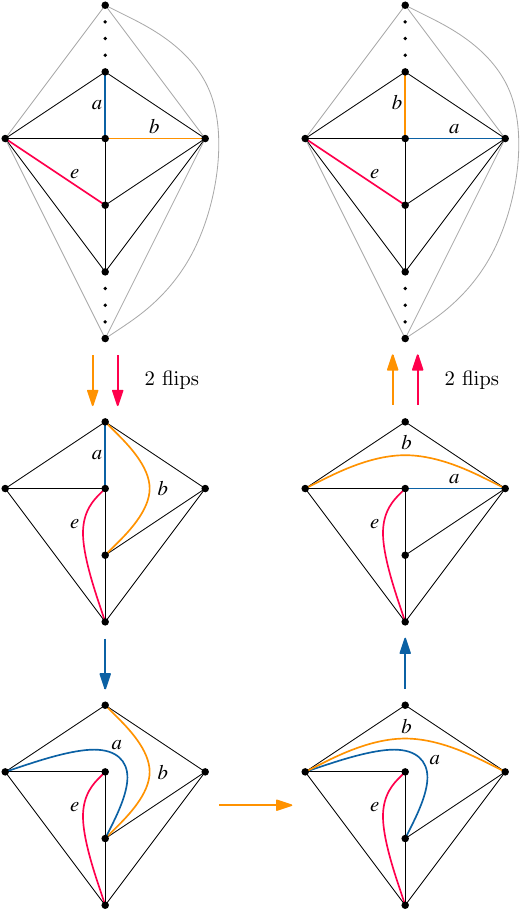}
   \caption{A sequence of seven flips that swaps two edges $a$ and $b$ that are consecutive around a vertex on the spine. Although edge $e$ ends up at the same place as at the start of the sequence, it essentially acts as a catalyst here. If we did not flip it, we would not be able to flip edge $a$ after edge $b$, as that would create a duplicate edge.}
   \label{fig:spineswap}
 \end{figure}

\begin{theorem}
\label{thm-upper-bound-comb}
 Any edge-labelled combinatorial triangulation with $n$ vertices can be transformed into any other by $O(n\log n)$ flips.
\end{theorem}
\begin{proof}
We show that we can transform any edge-labelled combinatorial triangulation into a canonical one using $O(n\log n)$ flips.
As flips are reversible, we can also go from the canonical triangulation to any other, which proves the theorem.

We use a canonical triangulation much like the one used by Sleator~\etal~\cite{STT92} for the vertex-labelled variant.
It is a double wheel: a cycle of length $n - 2$ (called the \emph{spine}), plus a vertex $\vI$ inside the cycle and a vertex $\vO$ outside the cycle, each connected to every vertex on the cycle (see Figure~\ref{fig:comb-canonical}).
For our canonical labelling, we separate the labels into three groups.
\prc{Group \gS contains labels $1$ through $n - 2$, which we place} on the spine edges, starting with the edge on the outer face and continuing in clockwise order around $\vI$.
The next $n - 2$ labels make up group \gI and are placed on the edges incident to $\vI$ in clockwise order, starting with the edge incident to the vertex shared by the edges with labels 1 and 2.
Finally, group \gO \prc{~consists} of the last $n - 2$ labels, which we place on the edges incident to $\vO$ in clockwise order, starting with the edge that shares a vertex with the edge labelled $2n - 4$.
 
Our algorithm first ignores the labels and transforms the given triangulation into the unlabelled canonical triangulation.
This requires $O(n)$ flips~\cite{STT92} and results in the correct graph, although the labels may be in arbitrary positions.
To fix the labels, we first get the groups right, so all labels in group $\gS$ are on the spine, etc., and then rearrange labels within each group.

We use two main tools for this.
The first is a \emph{swap} that interchanges one spine edge with an incident non-spine edge in seven flips, using the flip sequence depicted in Figure~\ref{fig:spineswap}.
Our second tool is a \emph{scramble} algorithm that reorders all labels incident to $\vI$ or $\vO$ using $O(n \log n)$ flips.
To do this we first flip the spine edge that is part of the exterior face (labelled 1 in Figure~\ref{fig:comb-canonical}) and then apply the algorithm from Theorem~\ref{thm-upper-bound} to the outerplanar graph induced by the spine plus $\vI$ (or $\vO$), observing that no flip will create a duplicate edge since the omitted edges are all incident to $\vO$ (resp. $\vI$).
Note that this method cannot alter the labels on the two non-spine edges that lie on the exterior face of the outerplanar graph (labelled 7 and 12 in Figure~\ref{fig:comb-canonical}), but since there are only two of these, we can move them to their correct places by swapping them along the spine, using $O(n)$ flips total.

To get the labels of group $\gS$ on the spine, we partner every edge incident to $\vI$ that has a label in $\gS$ with an edge on the spine that has a label in $\gI$ or $\gO$.
A scramble at $\vI$ makes each such edge incident to its partner, and then swaps partners.
By doing the same at $\vO$, all labels of $\gS$ are placed on the spine.
Next we partner every edge incident to $\vI$ that has a label in $\gO$ with an edge incident to $\vO$ that has a label in $\gI$.
A scramble at $\vI$ makes partners incident, and three swaps per pair then exchange partners.

This ensures that each edge's label is in the correct group, but the order of the labels within each group may still be wrong.
Rearranging the labels in $\gI$ and $\gO$ is straightforward, as we can simply scramble at $\vI$ and $\vO$, leaving only the labels on the spine out of order.
We then use swaps to exchange the labels on the spine with those incident to $\vI$ in $O(n)$ flips and scramble at $\vI$ to order them correctly.
Since this scramble does not affect the order of labels on the spine, we can simply exchange the edges once more to obtain the canonical triangulation.
\end{proof}

\prc{
\subsection{Lower bound}

The proof for the lower bound for combinatorial triangulations is very similar to the lower bound for triangulations of a convex polygon, described in Section~\ref{sec:el-convex-lb}. We again construct a graph grammar, which describes transformations on the dual graph that correspond to flips.

As our primary graph is a combinatorial triangulation, each vertex of the dual graph corresponds to a triangle and has degree three. As such, there is no distinction between internal nodes and leaves, and no root. This means that we need to adapt our definitions slightly. Without a root, the placement of the edge-end labels is less constrained. We only require that they occur in counter-clockwise order around each vertex. The order of labels in each tag can now follow the placement of the edge-end labels: the first label belongs to the primary edge corresponding to the dual edge with edge-end label 1, and so on.

Finally, we need a few more production \prc{rules} to deal with all possible rotations of the edge-end labels around the two triangle-vertices involved in the flip. The full collection of rules is shown in Figure~\ref{fig:el-combinatorial-graph-grammar}.

\begin{figure}[htbp]
 \centering
 \includegraphics{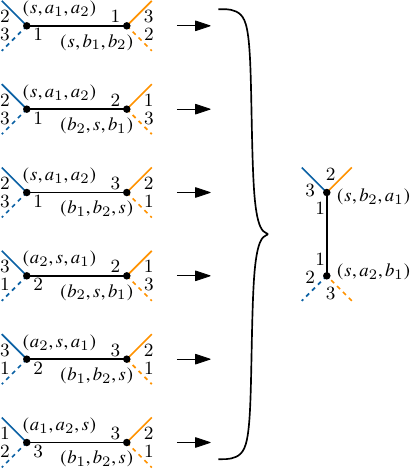}
 \caption{A graph grammar that corresponds to flips in edge-la\-belled combinatorial triangulations. The right-hand side of all productions is the same.}
 \label{fig:el-combinatorial-graph-grammar}
\end{figure}

As the dual graph of an $n$-vertex combinatorial triangulation has $2n - 4$ vertices, Theorem~\ref{thm:el-stt-reachable} gives us the following bound.

\begin{lemma}
 \label{lem:el-combinatorial-reachable}
 Given an $n$-vertex edge-la\-belled combinatorial triangulation $G$, the number of distinct edge-la\-belled triangulations reachable from $G$ in $m$ flips is at most $13^{2n - 4 + 2m}$.
\end{lemma}

Again, Sleator, Tarjan, and Thurston~\cite{STT92} show that this bound can be significantly reduced (to $3^{2n - 4}8^m$), but the simple bound suffices for our purposes.

\begin{theorem}
 \label{thm:el-combinatorial-lb}
 There are pairs of edge-la\-belled combinatorial triangulations with $n \geq 5$ vertices such that transforming one into the other requires $\Omega(n \log n)$ flips.
\end{theorem}
\begin{proof}
 If we fix the labelling of the spine edges in the canonical triangulation from the proof of Theorem~\ref{thm-upper-bound-comb}, any relabelling of the remaining edges is unique. Thus, there are at least $(2n - 6)!$ distinct edge-la\-belled combinatorial triangulations. Combined with Lemma~\ref{lem:el-combinatorial-reachable}, this implies that $13^{2n - 4 + 2d} \geq (2n - 6)!$, where $d$ is the diameter of the flip graph. We derive the following.
 \begin{align*}
  13^{2n - 4 + 2d}~~&\geq~~(2n - 6)! \\
  13^{2n - 4 + 2d}~~&\geq~~n!\text{\hspace{8em}(for $n \geq 5$)} \\
  \log_{13} 13^{2n - 4 + 2d}~~&\geq~~\log_{13} n! \\
  2n - 4 + 2d~~&\geq~~\log_{13} n! \\
  2d~~&\geq~~\log_{13} n! - 2n + 4 \\
  2d~~&\geq~~\Omega(n \log n) - O(n) \\
  d~~&\geq~~\Omega(n \log n) \qedhere
 \end{align*}
\end{proof}

\subsection{Simultaneous flips}

Recall that, in a triangulation of a convex polygon, a simultaneous flip is a set of flips that are executed in parallel, such that no two flipped edges share a triangle. In a combinatorial triangulation, we have the additional requirement that the resulting graph may not contain duplicate edges.

Simultaneous flips in combinatorial triangulations were first studied by Bose~\etal~\cite{Simultaneous06}. They showed a tight $\Theta(\log n)$ bound on the diameter of the flip graph. As part of their proof, they showed that every combinatorial triangulation can be made 4-connected with a single simultaneous flip. Recently, Cardinal~\etal~\cite{CHKTW15} proved that it is possible to find such a simultaneous flip that consists of fewer than $2n/3$ individual flips. They used this result to obtain arc drawings of planar graphs in which only $2n/3$ edges are represented by multiple arcs.

In this section, we show that, just as in the non-simultaneous setting, we obtain the same bounds for edge-la\-belled convex polygons and edge-la\-belled combinatorial triangulations. That is, we can transform any edge-la\-belled combinatorial triangulation into any other with $O(\log^2 n)$ simultaneous flips, and $\Omega(\log n)$ simultaneous flips are sometimes necessary. The lower bound holds already in the unlabelled setting, if one vertex has linear degree in the first triangulation, while every vertex has constant degree in the second. We now prove the upper bound.

\begin{figure}[htb]
 \centering
 \includegraphics{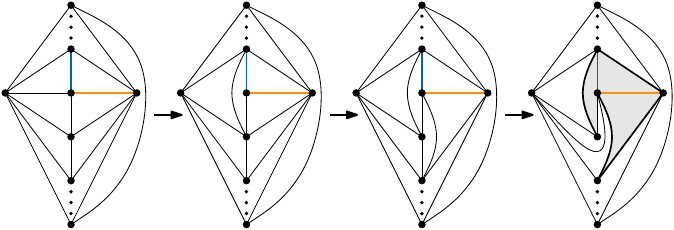}
 \caption{A sequence of three flips that creates a pentagon (shaded) in which the two highlighted edges can be swapped. All new edges and all diagonals of the pentagon are incident to one of the four spine vertices shown.}
 \label{fig:el-spineswap-simultaneous}
\end{figure}

\begin{theorem}
 Any edge-la\-belled combinatorial triangulation with $n$ vertices can be transformed into any other by $O(\log^2 n)$ simultaneous flips.
\end{theorem}
\begin{proof}
 We closely follow the strategy of the proof of Theorem~\ref{thm-upper-bound-comb}. We first transform the given triangulations into the canonical one with $O(\log n)$ simultaneous flips, using the result of Bose~\etal~\cite{Simultaneous06}. This reduces the problem to sorting the edge labels on the canonical triangulation. In the non-simultaneous setting, we did this by reordering the labels on the edges incident to $\vI$ or $\vO$ (called \emph{scrambling}), and swapping a subset of spine edges with incident non-spine edges. Thus, the theorem follows if we can show how to perform these operations with $O(\log^2 n)$ simultaneous flips.

 Since the sequence of flips from Figure~\ref{fig:spineswap}, which swaps a single spine edge with an incident non-spine edge, only involves a constant number of triangles, it is tempting to think we can simply perform many of these swaps simultaneously. Unfortunately, this is not the case, since the sequence creates the edge $(\vO, \vI)$. This means that trying to perform this sequence simultaneously in different locations would create a duplicate edge. Therefore we use a slightly longer sequence that creates a pentagon containing the edges to be swapped (illustrated in Figure~\ref{fig:el-spineswap-simultaneous}), performs the swap inside this pentagon, and restores the canonical triangulation, using a total of eleven flips. The crucial property of this sequence is that it only creates edges incident to four spine vertices near the edge to be swapped. Thus, we can perform any number of swaps simultaneously without creating duplicate edges, as long as each swap is at distance four or more from the others. This means that, given a set of spine edges to swap, we can divide them into four rounds such that the edges to be swapped in each round are at distance four or more, and perform the swaps in each round simultaneously. Thus, we can swap any subset of spine edges with $O(1)$ simultaneous flips.

 To scramble the edges incident on $\vO$, we first flip to create $(\vO, \vI)$ and then apply the algorithm from Theorem~\ref{thm:simultaneous} to the outerplanar graph induced by the edges incident to $\vO$. This uses $O(\log^2 n)$ simultaneous flips to rearrange all labels, except for those on the two outermost edges that are part of the boundary. In the non-simultaneous setting, we fixed this by swapping these labels along the spine, but this would take too many flips here. Instead, if the labels that need to be on the outermost edges are in the interior, we use Theorem~\ref{thm:simultaneous} to place these labels on the interior edges closest to the outermost edges. Then, we can exchange them with the labels on the outermost edges with only three swaps. This ensures that the outermost edges have the correct labels, so a second application of Theorem~\ref{thm:simultaneous} can place the remaining labels in the right order. If the label for one of the outermost edges is not in the interior and not already in place, it must be on the other outermost edge. In this case, we can first exchange it with the label on a nearby interior edge with a constant number of swaps. The entire sequence requires $O(\log^2 n)$ simultaneous flips.

 Since these operations use $O(1)$ and $O(\log^2 n)$ simultaneous flips, and we can sort the labels with a constant number of applications, the theorem follows.
\end{proof}
}

\section{Conclusions}

We have initiated the 
exploration of flips in edge-labelled triangulations and formulated the Orbit
Conjecture: that for any two edge-labelled triangulations of a 
polygon, point set, or vertex set \prc{(in the combinatorial setting)},
there is a sequence of flips that transforms one into the other if and only if for every label, the initial and final edge with that label lie in the same orbit.
Furthermore, we conjecture that the diameter of any connected component of the flip graph is bounded by a polynomial in the number of vertices.

We have established the conjecture for combinatorial triangulations and triangulations of convex polygons---all edges are in one orbit so the flip graph is connected, and its diameter is $\Theta(n \log n)$.
This means that the worst case number of flips is a $\log n$ factor more than for the unlabelled case.  
With simultaneous flips, edge labels raise the worst case bound
by at most a $\log n$ factor, but we could not prove that this is tight.
We also established the conjecture for spiral polygons, characterizing the orbits and showing that each connected component of the flip graph has diameter $\Theta(n^2)$.  
For general polygons there are examples of a connected flip graph of  diameter  $\Theta(n^3)$. 

%We conclude with the following open problem: Does the Orbit Conjecture hold for edge-labelled triangulations of simple polygons or edge-labelled triangulations of point sets?

\section*{References}

\bibliographystyle{abbrv}
\bibliography{cgta}

\end{document}